%% file: manuscript.tex
\documentclass[%
11pt,
reprint,
onecolumn,
tightenlines,
superscriptaddress,
preprintnumbers,
nofootinbib,
amsmath,amssymb,amsthm,
physrev,
eqsecnum,tikz,
]{revtex4-2}

\usepackage{hyperref} 

\input{preamble}

\begin{document}
	
\preprint{To appear in Journal of the Mechanics and Physics of Solids (\url{https://doi.org/10.1016/j.jmps.2024.105743}).}

\title{Nonuniqueness in Defining the Polarization: Nonlocal Surface Charges and the \\ Electrostatic, Energetic, and Transport Perspectives}

\author{Shoham Sen}
    \email{shoham.sen16@gmail.com}
    \affiliation{Department of Civil and Environmental Engineering, Carnegie Mellon University}

\author{Yang Wang}
    \affiliation{Pittsburgh Supercomputing Center}

\author{Timothy Breitzman}
    \affiliation{Air Force Research Laboratory}
    
\author{Kaushik Dayal}
    \affiliation{Department of Civil and Environmental Engineering, Carnegie Mellon University}
    \affiliation{Center for Nonlinear Analysis, Department of Mathematical Sciences, Carnegie Mellon University}
    \affiliation{Department of Mechanical Engineering, Carnegie Mellon University}

\date{\today}

\begin{abstract}
    Ionic crystals, ranging from dielectrics to solid electrolytes to complex oxides, play a central role in the development of modern technologies for energy storage, sensing, actuation, and other functional applications. 
    Mesoscale descriptions of these crystals are based on the continuum polarization density field to represent the effective physics of charge distribution at the scale of the atomic lattice. 
    However, a long-standing difficulty is that the classical electrostatic definition of the macroscopic polarization --- as the dipole or first moment of the charge density in a unit cell --- is not unique; rather, it is sensitive to translations of the unit cell in an infinite periodic system. 
    This unphysical non-uniqueness has been shown to arise from starting directly with an infinite system --- wherein the boundaries are ill-defined --- rather than starting with a finite body and taking appropriate limits.
    This limit process shows that the electrostatic description requires not only the bulk polarization density, but also the surface charge density, as the effective macroscopic descriptors; that is, a nonlocal effective description.
    Other approaches to resolve this difficulty include the popular modern theory of polarization that completely sets aside the polarization as a fundamental quantity in favor of the change in polarization from an arbitrary reference value and then relates the change in polarization to the transport of charge (the ``transport'' definition); or, in the spirit of classical continuum mechanics, to define the polarization as the energy-conjugate to the electric field (the ``energetic'' definition).

    This work examines the relation between the classical electrostatic definition of polarization, and the transport and energy-conjugate definitions of polarization.
    We show the following:
    (1) The transport of charge does not correspond to the change in polarization in general; instead, one requires additional simplifying assumptions on the electrostatic definition of polarization for these approaches to give rise to the same macroscopic electric fields. Thus, the electrostatic definition encompasses the transport definition as a special case.
    (2) The energy-conjugate definition has both bulk and surface contributions; while traditional approaches neglect the surface contribution, we find that accounting for the nonlocal surface contributions is essential to be consistent with the classical definition and obtain the correct macroscopic electric fields.
\end{abstract}
\maketitle

\section*{Highlights}
\begin{itemize}
    \item Presents surface charge density as an accompaniment to bulk polarization density to resolve the non-uniqueness in the classical electrostatic definition of polarization.
    \item Connects the electrostatic definition of polarization to the Modern Theory of Polarization's transport-based definition of polarization. 
    \item Shows a corresponding surface free energy that goes with the bulk free energy to connect the thermodynamic/energetic definition of polarization with the electrostatic definition of polarization.
    \item Presents a zero-temperature rigorous proof of how polarization can be derived from quantum mechanics using homogenization.
\end{itemize}


\section{Introduction}

Dielectric materials, ranging from solid electrolytes to complex oxides to dielectric elastomers, are central to modern technologies for energy storage, sensing, actuation, and other functional applications.
Electrostatic (Coulombic) interactions have a central role in the structure and response of such materials.
These interactions arise at the atomic scale in the interaction between electrons and protons; consequently, it is important to begin at the atomic or quantum scales when modeling these materials.
On the other hand, it is also important for models to be applicable at much larger macroscopic scales.
Hence, an important area of current activity is multiscale analyses that bridge these length scales, e.g. [\cite{grasinger2020statistical, grasinger2021architected, grasinger2021flexoelectricity, grasinger2022statistical,torbati2022coupling, liu2018emergent, chen2021interplay, ahmadpoor2013apparent, liu2013flexoelectricity, james-muller, sharp1994electrostatic, xiao2005influence,steigmann2018mechanics}]. 

A central quantity of interest is the polarization vector --- i.e., the dipole density per unit volume --- that serves as the key multiscale mediator of the electrostatic interactions.
The polarization provides a macroscopic quantity that captures the key aspects of the molecular-scale electrostatic interactions.
However, an important fundamental issue in the atomic and quantum scale modeling of these materials is that, in a periodic crystal, the polarization depends on the choice of the unit cell.
That is, the polarization computed using two choices of unit cell that are related by a simple translation --- which is physically irrelevant in a periodic system --- can differ significantly, including by a complete flip in the direction.

The dependence of the polarization on the choice of unit cell has been recognized for several decades.
A widely used resolution to this issue is through the ``Modern Theory of Polarization'' [\cite{resta2007theory}], which treats the change in polarization as a fundamental quantity and relates this to the transport of charge.
Another approach, rooted in thermodynamics [\cite{toupin1956elastic}], defines the polarization as the energy-conjugate to the electric field.
The aim of this paper is to elucidate the connections between these seemingly disparate perspectives.

\paragraph*{Prior Work.}

The polarization vector field is a key element of numerous successful macroscopic continuum models of dielectric and functional materials: e.g.,  piezoelectricity [\cite{haeni2004room,tagantsev1986piezoelectricity,tagantsev2010domains,rahmati2019nonlinear,deng2017continuum}], flexoelectricity [\cite{zubko2013flexoelectric,abdollahi2015revisiting,abdollahi2014computational,abdollahi2019converse,abdollahi2015fracture,liu2013flexoelectricity,krichen2016flexoelectricity}], ferroelectricity [\cite{Irene1,yang2011completely,yang2011effect,yang2012free,yang2012influence,yang2012microstructure,yen2008study,shu2008constrained}], dielectric elastomers [\cite{khandagale2023statistical, khandagale2024statistical,grasinger2020statistical,grasinger2021architected, grasinger2021flexoelectricity, grasinger2022statistical,friedberg2023electroelasticity,itskov2019electroelasticity}], and analogous models in micromagnetism [\cite{james1990frustration,james1998magnetostriction,james2000martensitic,desimone2002constrained,tickle1999ferromagnetic,benesova2018existence,wang2001gauss,garcia2003accurate,bosse2007microdomain}].
Starting from the atomic scale, several multiscale approaches have used the polarization as a means of coarse-graining, e.g. [\cite{james1994internal,jha2023discrete,jha2023atomic,tadmor-EffHamil,grasinger2021flexoelectricity,grasinger2021architected,grasinger2020statistical,marshall2014atomistic,cicalese2009discrete,bach2020discrete,alicandro2008continuum}].

An important area of current research is in embedding, i.e., using expensive but accurate atomic or quantum models in limited regions where they are required, and embedding these regions in a coarse-grained continuum description  [\cite{QC-review1,QC-review2,tadmor2011modeling,kulkarni2008variational,knap-ortiz,qc-original, miller2002quasicontinuum, kochmann2016quasicontinuum,makridakis2014atomistic}].
While these have largely focused on models without electrostatics, recent works have extended these to account for electrostatics, and the polarization plays a central role in the coarse-graining [\cite{jha2023atomic,marshall2014atomistic,garcia2007efficient,garcia2009linear}]. 

Turning specifically to prior work on finding a suitable definition of polarization, there has been significant contributions from the quantum mechanics community. Classically, the formal definition of polarization is the dipole moment per unit volume; we will refer to this as the classical/electrostatic definition of polarization denoted by $\ve{p}_{cl}$; see Figure \ref{Fig.DomHom} and subsequent discussion. 
However, for infinite periodic systems, $\ve{p}_{cl}$ depends on the choice of unit cell, making it non-unique. 
The non-uniqueness in the classical definition gave rise to the Modern Theory of Polarization pioneered by Resta and Vanderbilt [\cite{resta2007theory}], building on contributions by [\cite{ortiz1994macroscopic,martin1974comment,resta1993macroscopic,king1993theory,littlewood1980calculation}]. 
They posit that since experiments measure the change in polarization rather than the polarization itself, the change in polarization is the quantum mechanical observable. 
Further, since a change in polarization causes a flow of current, we can measure the change in polarization by measuring the current flow. 
Assuming adiabatic changes, they relate the current to the Berry phase of the corresponding quantum mechanical wavefunction. 
The resulting polarization will be labeled $\ve{p}_{T}$. 
Lastly, there is the energetic definition of polarization [\cite{toupin1956elastic}]: $\ve{p}_{W}:=\dfrac{\partial W}{\partial \ve{E}}$ is the energy-conjugate to the electric field $\ve{E}$ where $W$ is the appropriate free energy density.
This last definition depends on the model chosen to describe the system through the choice of $W$.

Landauer [\cite{landauer1981pyroelectricity}] discusses the possibility of a surface-level description, in addition to the bulk polarization, for piezo- and pyro-electrics.
Mathematical approaches to related problems typically begin with a finite body and take the limit; such approaches can naturally account for surface effects.
For instance, [\cite{james1994internal,muller2002discrete,schlomerkemper2009discrete,jha2023discrete}] take such an approach; however, they begin with microscopic dipoles rather than a charge distribution and hence do not have surface effects.
Related work in [\cite{rosakis2012continuum}] with short-range atomic forces gives rise to a surface energy.
In [\cite{marshall2014atomistic,jha2023atomic}], they formally study such limits using distributions of charges and find that the coarse-grained description requires both a polarization --- defined as the dipole moment of the charge over any given periodic unit cell --- and the corresponding surface charge density for the same given unit cell.

\paragraph*{Contributions of This Work.}

This work relates the three definitions of polarization, and shows that the ``classical'' definition encompasses the ``transport'' and ``energetic'' definitions of polarization.
By examining the continuum limit of when the lattice spacing is much smaller than the characteristic dimensions of the body, we show that accounting for the boundaries consistently provides a route to uniquely compute electric fields and potentials despite the non-uniqueness of the polarization. Specifically, different choices of the unit cell in the interior of the body leads to correspondingly different partial unit cells at the boundary; while the interior unit cells satisfy charge neutrality, the partial cells on the boundary need not.
The net effect is that the bulk and surface contributions are both dependent on the choice of unit cell, but balance each other to give the same electric fields for any choice of unit cell.

First, we consider the relation between the classical and the transport definitions.
The transport definition computes the change in polarization by integrating the current over time. 
In Theorem \ref{Thm4.1}, we show that the current calculated in the transport definition and the electrostatic definition of polarization give rise to the same electrostatic potential for the system, provided that the change in the system is slow with respect to the driving. 
As mentioned in Remark \ref{ReMark4.3}, many ferroelectric phase transformations have a discontinuous change in polarization, making it hard to justify the adiabatic change hypothesis. 
Examples \ref{Ex-4.1} and \ref{ex2} show how the transport definition does not match the classical electrostatic definition of polarization in general, but does so only under certain strict assumptions.
Though the transport definition provides an unambiguous quantity that is independent of the choice of unit cell, it does not give rise to the correct electrostatic fields except under these assumptions.

Second, we consider the relation between the classical and the energetic definitions.
We show that the free energy density definition has two parts, a surface and a bulk part, analogous to the surface and the bulk description in the electronic definition of polarization.
While traditional approaches neglect the nonlocal surface contribution, we show that it is essential to account for both these parts to be consistent with the electrostatic definition and give rise to the correct electrostatic fields.

It is worth highlighting an important advantage of being able to use the classical definition of polarization, rather than the Berry phase approach typically used in the transport definition of polarization.
The classical definition requires only the charge distribution, rather than the transport definition which requires the quantum mechanical wavefunction in its typical formulation in terms of the Berry phase.
The charge distribution is an outcome of most standard atomic or quantum models, whereas the (true) wavefunction is typically not obtainable even in common electronic structure methods such as density functional theory (DFT).

\paragraph*{Organization.}

In Section \ref{cl}, we treat the classical definition of polarization. We start with a simple example explaining the non-uniqueness plaguing $\ve{p}_{cl}$, followed by describing our real space approach to polarization; we address the non-uniqueness problem here. We formulate the problem mathematically and follow this with the homogenized result. While we do not go into the details of deriving the result, we explain the homogenization process with simple and illustrative 1-d examples. We then present the polarization problem for an infinite crystal as a limit of a finite crystal.

This is followed by the transport definition of polarization in Section \ref{Berry}. We start by presenting the motivation and then simple calculations showing that the classical and transport definitions are not generally equivalent. 
Since the transport definition of polarization is posed for infinite crystals, we comment on approaches taken to present this for finite crystalline domains, referencing opposing views by other researchers. The main result of the section connects $\ve{p}_{cl}$ to $\ve{p}_{T}$.

In Section \ref{W}, we analyze the free energy density definition of polarization; differentiating the free energy with respect to the electric field. 
In Section \ref{VariationalDerivative}, we find the functional/variational derivative of the free energy with respect to the external potential, obtaining a relationship to connect $\ve{p}_{cl}$ to $\ve{p}_{W}$. 
We then consider a finite crystalline domain and homogenize the equation above, properly scaled, by following a procedure similar to that used in [\cite{PolarizationTheorem}].
The details of the calculation can be found in Section \ref{Polarization from Schrodinger's equations}, but they are separated from Section \ref{VariationalDerivative} as they are not essential for the physical interpretation. 
This result helps establish the connection between the electrostatic and energetic definitions of polarization. 
Both definitions of polarization have an associated surface charge distribution that is required for uniqueness.

\subsection{Notation.}
\label{Notation}
We use standard notations from analysis. $L^p$, with $1\leq p<\infty$, denotes the space of functions whose $p$-{th} power is integrable. $\scr{C}^k$, with $k\geq0$ denotes the space of functions with continuous $k$ derivatives. $W^{1,p}$ denote the Sobolev space of functions where the function and its weak first derivative are $p$ integrable. The space $\scr{C}_c$ denotes a continuous function with compact support.

We often use the notation $f\vert_{\Omega}$ to denote the values of the function restricted to the subset $\Omega$. The measure for volume integration is $\dm V_{\ve{x}}$ while that for surface integration is $\dm S_{\ve{x}}$, where the subscript describes the variable. We will use the shorthand $f(x,\cdot)$ for a function $f(x,y)$ to denote that $x$ is fixed and that the variable replaced with $\cdot$ is allowed to vary. We will use $\Delta$ to denote the Laplacian operator.

We define the Minkowski sum between two sets $P$ and $Q$ as 
\begin{equation}
P\oplus Q:=\{p+q\vert p\in P,~q\in Q\}
\end{equation}
where for our purposes, we require that $P,Q\subset\Scr{R}^d$ and we use $x$ to denote the singleton set $\{x\}$, provided no confusion arises.

We define a lattice $\scr{L}$ for this problem,
\begin{equation}
\scr{L}:=\big\{\ve{x}\in\Scr{R}^3\big\vert \ve{x}=\sum_{i}x^i\ve{e}_i,~x^i\in\Scr{Z},~i=1,2,3\big\}
\end{equation}
where $\{\ve{e}_i\}_1^3$ are linearly independent lattice vectors describing the lattice. The corresponding dual lattice ($\scr{RL}$) and dual vectors $\{\ve{e}^j\}_1^3$ are defined by the orthogonality relation $\ve{e}_i\cdot\ve{e}^j=\delta^j_{i}$. We scale $\scr{L}$ by a factor $l\in(0,1]$ and obtain the scaled lattice $l\scr{L}$.

When homogenizing, we want to do so over some domain $\Omega\subset\Scr{R}^3$. For Section \ref{cl}, this happens to be the support of the charge distribution, ie, the charge distribution is zero outside $\Omega$. The following definitions are with respect to $\Omega$.

We also use the following definitions related to the lattice:
\begin{enumerate}
    \item $\Box$ -- denotes a choice of unit cell.
    \item $l\Box$ -- a scaled unit cell.
    \item $\lr$ -- a partial unit cell.
    \item $l\lr$ -- a scaled partial unit cell. A partial unit cell is a unit cell through which the boundary passes.
    \item $\hat{\ve{x}}_l$ -- is the corner map and denotes a specific point with respect to a unit cell. This way, a specific scaled unit cell (partial unit cell) can be referred to as $\hat{\ve{x}}_l\oplus l\Box$ ($\hat{\ve{x}}_l\oplus l\lr$). For simplicity, one may think of the corner map varying over the lattice points, though that restriction is unnecessary. Appendix \ref{Admissible set of unit Cells} presents a rigorous definition for the above formally defined quantities.
    \item $\Omega(\Box)$ -- denotes the collection of scaled unit cells completely contained inside $\Omega$.
    \item $\Omega(\lr):=\Omega\backslash\Omega(\Box)$ -- denotes the collection of partial unit cells.
\end{enumerate}
Refer to Figure \ref{Fig.DomHom} for a visual representation.

Note that if $\ve{x}\in\Omega$, then $\hat{\ve{x}}_l$ may lie outside $\Omega$. Thus, we will use $\hat{\ve{x}}_l\in\Omega(\Box)$ and $\hat{\ve{x}}_l\in\Omega(\lr)$ as short hand to imply $\hat{x}_l\oplus l\Box\subset\Omega$ and $\hat{x}_l\oplus l\Box\cap\Omega^c\neq\{\}$ respectively. 
We use $\Omega^c$ to describe the complement/absence of the domain. 
The characteristic function is denoted as $\mathbbm{1}_{\Omega}$.

Defining ``bulk charge neutrality'' for a given charge distribution in a domain will be convenient. The rigorous definition is available in  Appendix \ref{Bulk Charge Neutrality}. This definition ensures that the system can only be described with dipoles in the bulk. To get a clearer understanding, consider the weaker assumption of ``charge neutral unit cells''. This means that the total charge in each unit cell is 0; these are unit cells contained entirely in the domain $\Omega$. Thus, the net electrical information in each unit cell is described using dipoles. Since we want to describe the bulk using polarizations, we use the general condition of ``bulk charge neutrality''.  Throughout this paper, we have set $\epsilon_0=1$; this quantity measures the electrostatic interaction strength. When taking the continuum limits, we want to fix the interaction strength between two atoms.

\section{Classical Definition of Polarization}\label{cl}

We briefly summarize key results about the classical definition of polarization from [\cite{PolarizationTheorem}]; the rigorous proofs can be found there. 
The definition of polarization used here is the dipole moment per unit volume, calculated as the average dipole moment per unit cell. 
The origin of the non-uniqueness is the fact that the dipole moment depends on the choice of unit cell. 

In Example \ref{Ex-2.1}, we provide an illustration demonstrating the non-uniqueness in polarization. Subsequently, to elucidate the polarization theorem and its application to deriving the homogenized PDE, we present Examples \ref{Ex-2.2} and \ref{Ex-2.3}. In Example \ref{Ex-2.2}, we examine a standard 1-d electrostatics problem with specified Dirichlet boundary conditions. In Example \ref{Ex-2.3}, we examine a problem with mixed boundary conditions. We elaborate on these two examples to show the process of calculating the polarization and surface charge, ultimately demonstrating the uniqueness of the electrostatic potential.

\begin{example}\label{Ex-2.1}([\cite{resta2007theory,marshall2014atomistic}])
    Consider the charge distribution $\rho_l\colon\Scr{R}^3\to\Scr{R}$ given by $\rho_l(\ve{x})=\dfrac{1}{l}\sin\left(2\pi\dfrac{x_1}{l}\right)$; the charge distribution is $l$-periodic. Now we evaluate the dipole with respect to an arbitrary unit cell at $\ve{a}\oplus l[0,1]^3$:
    \beol{
        \ve{p}(\ve{a})=\dfrac{1}{l^3}\int_{\ve{a}\oplus l[0,1]^3}\dfrac{(\ve{x}-\ve{a})}{l}\rho_l(\ve{x})\dm V_{\ve{x}}=\int_{[0,1]}\ve{y}\sin\left(2\pi\dfrac{a_1}{l}+2\pi y_1\right)\dm y_1=-\dfrac{1}{2\pi}\cos{\left(2\pi\dfrac{a_1}{l}\right)}\bfe_1.
    }
    using $\ve{y} = \frac{\ve{x} - \ve{a}}{l}$.
    
    Thus, $\ve{p}$ depends on the choice of unit cell, i.e., the choice of $\ve{a}$, for a given charge distribution.
\end{example}

In [\cite{PolarizationTheorem}], we show that even though the polarization is not a unique quantity, it is accompanied by a surface charge density $\sigma$, and the polarization and $\sigma$ together maintain the uniqueness of the electrostatic potential. 
The definition of polarization is the dipole moment per unit cell of the lattice, while the surface charge $\sigma$ is the total charge in the partial unit cell.

We start with a sequence of electrostatic Poisson equations posed on $\Scr{R}^3$ with respect to a sequence of ``bulk charge neutral'' charge distributions $\dfrac{\tilde{\rho}_l}{l}\colon\Omega\subset\Scr{R}^3\to\Scr{R}$ with $||\tilde{\rho}_l(\hat{\ve{x}}+l\cdot)||_{L^2(\Box)}<\infty$.
The corresponding electric potential field is denoted by $\Phi_l(\ve{x})$.
\begin{equation}
\begin{split}
-\Delta &\Phi_l(\ve{x})=\dfrac{\tilde{\rho}_l(\ve{x})}{l} \text{ in } \Scr{R}^3\backslash\Gamma_d,\\
	 &\Phi_l=f \text{ on } \Gamma_d,\\
	 &\Phi_l=0 \text{ as } |\ve{x}|\to\infty,
\end{split}
\label{Prob}
\end{equation}
$\Gamma_d\subset\partial\Omega$ is the portion of the boundary with Dirichlet boundary conditions specified, corresponding to electrodes.

If we now take the limit as $l\to0$ for the problem in \eqref{Prob}, then we obtain the following homogenized problem:
\begin{equation}
\begin{split}
    \Delta_{\ve{x}}&\Phi_0(\ve{x})=\mathbbm{1}_{\Omega}\divergence(\ve{p}_0) \text{ in } \Scr{R}^3\backslash\Gamma_d,\\
	 &\Phi_0=f \text{ on } \Gamma_d,\\
	 &\left\llbracket\dfrac{\partial\Phi_0}{\partial \bfn}\right\rrbracket=\ve{p}_0\cdot \ve{n}+\sigma_p \text{ on } \partial\Omega\backslash\Gamma_d,
\end{split}
\label{Result}
\end{equation}
where in the above, the subscript $\ve{x}$ on the Laplacian denotes that it acts on $\ve{x}$ and not on the microscopic variable $\ve{y}$.

Our calculations in \textbf{[}\cite{PolarizationTheorem}\textbf{]} are motivated by an earlier calculation by \textbf{[}\cite{marshall2014atomistic}\textbf{]}. We present a formal version of the calculation in \textbf{[}\cite{marshall2014atomistic}\textbf{]} to provide a glimpse for interested readers. Consider \eqref{Prob} with $|\Gamma_d|=0$. For this case, we can write the solution as:
\beol{\Phi_l(\ve{x})=\int_{\Omega} G(\ve{x},\ve{x}')\dfrac{\rho(\ve{x}')}{l}\dm V_{\ve{x}'},\label{formal}}
where $G(\ve{x},\ve{x}')$ is the free space Greens function. We break up the domain $\Omega$ into full and partial unit cells. Thus the integration can be represented by:
\beol{\Phi_l(\ve{x})=\sum_{\Omega(\Box)}\int_{l\Box} G(\ve{x},\ve{x}')\dfrac{\rho(\ve{x}')}{l}\dm V_{\ve{x}'}+\sum_{\Omega(\lr)}\int_{l\lr} G(\ve{x},\ve{x}')\dfrac{\rho(\ve{x}')}{l}\dm V_{\ve{x}'},}
Now we substitute, $\ve{x}'=\ve{\hat{x}}'+l\ve{y}'$, and use a Taylor expansion of the Greens function to get:
\bml{\Phi_l(\ve{x})&=\sum_{\ve{\hat{x}}'\in\Omega(\Box)}l^3\int_{\Box} G(\ve{x},\ve{\hat{x}}'+l\ve{y}')\dfrac{\rho(\ve{\hat{x}}'+l\ve{y}')}{l}\dm V_{\ve{y}'}+\sum_{\ve{\hat{x}}'\in\Omega(\lr)}l^3\int_{\lr} G(\ve{x},\ve{\hat{x}}'+l\ve{y}')\dfrac{\rho(\ve{\hat{x}}'+l\ve{y}')}{l}\dm V_{\ve{y}'}\\
    &\approx\sum_{\ve{\hat{x}}'\in\Omega(\Box)}l^3\int_{\Box} \Big[G(\ve{x},\ve{\hat{x}}')+l\ve{y}'\cdot\partial_{\ve{x}'}G(\ve{x},\ve{x}')\Big]\dfrac{\rho(\ve{\hat{x}}'+l\ve{y}')}{l}\dm V_{\ve{y}'}+\sum_{\ve{\hat{x}}'\in\Omega(\lr)}l^2\int_{\lr} G(\ve{x},\ve{\hat{x}}')\rho(\ve{\hat{x}}'+l\ve{y}')\dm V_{\ve{y}'},\\
    &=\sum_{\ve{\hat{x}}'\in\Omega(\Box)}l^3 \partial_{\ve{x}'}G(\ve{x},\ve{x}')\cdot\int_{\Box} \ve{y}'\rho(\ve{\hat{x}}'+l\ve{y}')\dm V_{\ve{y}'}+\sum_{\ve{\hat{x}}'\in\Omega(\lr)}l^2G(\ve{x},\ve{\hat{x}}')\int_{\lr} \rho(\ve{\hat{x}}'+l\ve{y}')\dm V_{\ve{y}'},\\
    &=\sum_{\ve{\hat{x}}'\in\Omega(\Box)}l^3 \partial_{\ve{x}'}G(\ve{x},\ve{x}')\cdot\ve{p}(\hat{\ve{x}}')+\sum_{\ve{\hat{x}}'\in\Omega(\lr)}l^2G(\ve{x},\ve{\hat{x}}')\sigma(\ve{\hat{x}}'),}
where, to get the penultimate line, we have used that each unit cell is charge neutral.
Now for $l$ small enough, we can replace the sum with an integration. This gives us,
\bml{\Phi_l(\ve{x})&=\int_{\Omega} \partial_{\ve{x}'}G(\ve{x},\ve{x}')\cdot\ve{p}({\ve{x}}')\dm V_{\ve{x}'}+\int_{\partial\Omega}G(\ve{x},\ve{x}')\sigma(\ve{x}')\dm S_{{\ve{x}'}}\\
                 &=\int_{\Omega} G(\ve{x},\ve{x}')\big(-\text{div}_{\ve{x}'}\ve{p}(\ve{x}')\big)\dm V_{\ve{x}'}+\int_{\partial\Omega}G(\ve{x},\ve{x}')\big(\ve{p}\cdot \ve{n}+\sigma(\ve{x}')\big)\dm S_{{\ve{x}'}}.}
Comparing the above with \eqref{formal} gives us the corresponding result for \eqref{Result}.

The derivation is available in [\cite{PolarizationTheorem}].  We write down the polarization theorem below,
\beol{
    \lim_{l\to 0}\int_{\Omega}\varphi(\ve{x})\dfrac{\tilde{\rho}_l(\ve{x})}{l}\dm V_{\ve{x}}=\int_{\Omega}\nabla_{\ve{x}}\varphi(\ve{x})\cdot\ve{p}_0(\ve{x})\dm V_{\ve{x}}+\int_{\partial\Omega}\sigma(\ve{x})\varphi(\ve{x})\dm S_{\ve{x}}\label{polthm}~~\forall\varphi\in L^2(\Omega)\cap L^{\infty}(\partial\Omega),
}
where $\tilde{\rho}_l$ satisfies the hypothesis above. Here, $\ve{p}_0$ denotes the polarization, and $\sigma$ denotes the surface charge; the expressions are given below.

The fact that the potential is unique follows from the condition that $\Phi_l$ converges weakly to $\Phi_0$ ($\Phi_l\rightharpoonup\Phi_0$). Moreover, since $\Phi_0$ is unique, we obtain that $\divergence(\ve{p}_{0})$ and $\sigma_p+\ve{n}\cdot\ve{p}_{0}$ are unique. Note that $\divergence\ve{p}_0=-\rho_b$ is the bound bulk charge density $\sigma_p$ is the partial charge density and $\ve{n}\cdot\ve{p}_{0}=\sigma_b$ is the bound surface charge density. 

We summarize the homogenization process below. 
Following Figure \ref{Fig.DomHom}, given a domain $\Omega\subset\Scr{R}^3$, we insert a lattice $\scr{L}$ scaled with a parameter $l\in(0,1]$. 
Now consider a sequence of ``bulk charge neutral'' charge distributions $\tilde{\rho}_l/l\colon\Omega\to\Scr{R}$ with bulk charge neutrality defined with respect to the lattice $l\scr{L}$. 
We decompose the domain $\Omega$ into $\Omega(\Box)\cup\Omega(\lr)$, the former composed of all the complete unit cells while the latter is composed of all the partial unit cells. 
The complete unit cells are highlighted in green, and the partial unit cells are highlighted in blue. In each of these complete unit cells, we replace the charge with the corresponding average dipole in the unit cell:
$$
    \ve{p}_l(\ve{x})=\dfrac{1}{|l\Box|}\int_{\hat{\ve{x}}\oplus l\Box}\dfrac{(\ve{z}-\hat{\ve{x}})}{l}\tilde{\rho}_l(\ve{z})\dm V_{\ve{z}}\mathbbm{1}_{\hat{\ve{x}}\oplus l\Box}(\ve{x}),
$$
where $\hat{\ve{x}}$ has been defined in Section \ref{Notation}.
In each of the partial unit cells, we calculate the total charge in each partial unit cell: 
$$\tilde{\sigma}_{p,l}=\int_{\hat{\ve{x}}\oplus l\lr}\dfrac{\tilde{\rho}_l(\ve{x})}{l}\dm V_{\ve{x}},$$
and the scaled area of the partial unit cell:
$$K_l=\dfrac{|\hat{\ve{x}}\oplus l\Box\cap\partial\Omega|}{l^{2}}=\dfrac{|\hat{\ve{x}}\oplus l\lr\cap\partial\Omega|}{l^{2}}.$$

\begin{figure}[htb!]
  \centering
  \includegraphics[width=\textwidth]{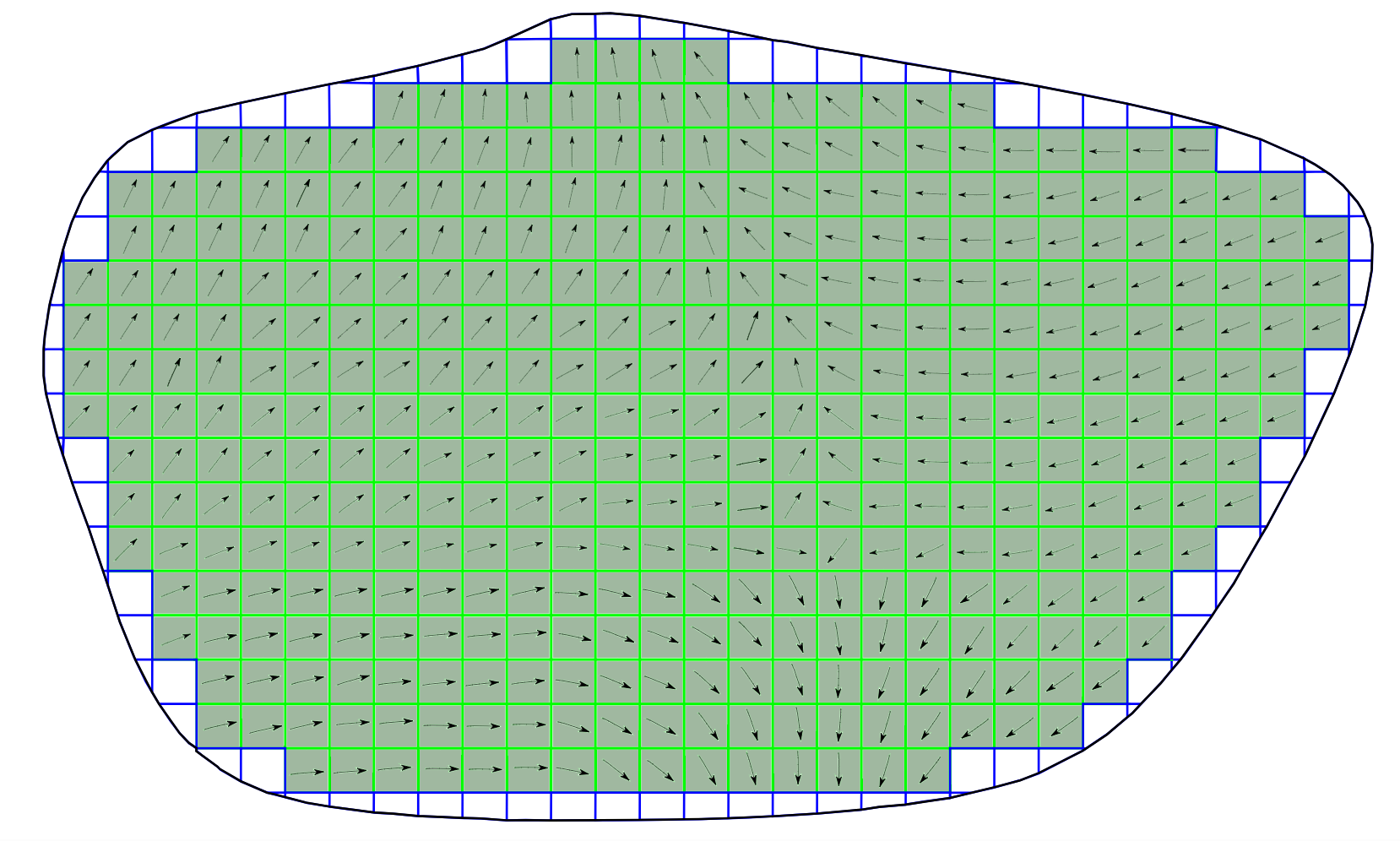}
  \caption{A domain $\Omega$ with ``bulk charge neutral'' locally periodic charge distribution in each unit cell. The green outline denotes a full charge-neutral unit cell, the blue outline denotes a partial unit cell, and the black outline marks $\partial\Omega$. The collection of full unit cells (in green) is $\Omega(\Box)$ while the collection of partial unit cells (in blue) is $\Omega(\lr)$. The arrows in the figure denote the dipoles in each unit cell. The homogenization process begins with inserting a (scaled) lattice, identifying the unit cells and partial unit cells, and replacing the charge in each unit cell with the average dipole moment. The charge in a partial unit cell gets mapped to the corresponding part of the boundary shared by that unit cell; the partial unit cells are left charge-less.}
  \label{Fig.DomHom}
\end{figure}

The surface charge density, associated with the corresponding corner of the partial unit cell, is defined on $\partial\Omega$ by:
$$
    \sigma_{p,l}(\hat{\ve{x}})=\dfrac{\tilde{\sigma}_{p,l}(\hat{\ve{x}})}{K_l(\hat{\ve{x}})}.
$$
Thus the homogenized problem is posed on $\Omega$ with zero charge distribution in $\Omega(\lr)$; the charge in the partial unit cells has been accounted for as a surface charge and distributed onto $\partial\Omega$.
In the above equation, the limit of $\sigma_{p,l} \to \sigma_p$ as $l\to 0$. Note that for $l$ small enough, the continuity of the potential implies $\Phi\vert_{\partial\Omega}\approx\Phi\vert_{\partial\Omega(\Box)}$.

To put this result into perspective, we discuss a few simple 1-d examples to explain the homogenization and limiting process.

\begin{figure}[htb!]
    \begin{center}
    \includegraphics[scale=0.74]{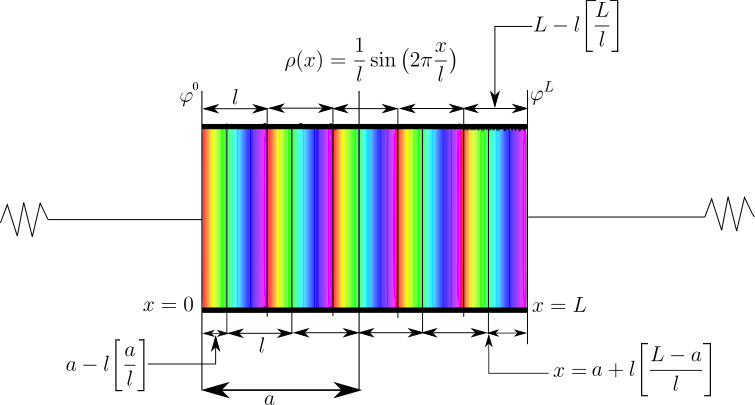}
    \caption{Image of a sinusoidal charge distribution with color used to denote the value of the charge. The markers on the two ends denote the $x$ values ($x=0,x=L$) and the potential at the two ends. Two choices of the lattice have been marked. The first choice of the origin of the lattice is at $x=0$; this has no partial unit cell on the left with a $[l\Big[\frac{L}{l}\Big],L]$ unit cell on the right. The next choice of the origin for the lattice is at $x=a$; this has $[0,a-l\Big[\frac{a}{l}\Big]]$ as the partial unit cell on the left and a $[a+l\Big[\frac{L-a}{l}\Big],L]$ partial unit cell on the right. We mark the two choices above and below the domain to avoid confusion.}
    \end{center}
    \label{Fig.PolProc}
\end{figure}

\begin{example}\label{Ex-2.2}
    Consider Fig. \ref{Fig.PolProc}, which shows a charge distribution $\rho\colon\Scr{R}\to\Scr{R}$ with the expression $\rho(x)=\frac{1}{l}\sin{\Big(2\pi\frac{x}{l}\Big)}$ supported on a domain $[0,L]$. 
    The electric potential $\varphi_l(x)$ is the solution of:
    \bml{
        &\varphi_l''(x)=-\rho(x)\mathbbm{1}_{[0,L]}(x) \text{ in } \Scr{R}\\
         &\varphi_l(0)=\varphi^0\\
         &\varphi_l(L)=\varphi^L,
    }
    The solution is:
    \beol{
        \varphi_l(x)=
        \begin{cases} 
                  0 & x< 0 \\
                  \varphi^0\left(1-\dfrac{x}{L}\right)+\varphi^L\left(\dfrac{x}{L}\right)-\dfrac{l}{4\pi^2}\dfrac{x}{L}\sin{\Big(2\pi\dfrac{L}{l}\Big)}+\dfrac{l}{4\pi^2}\sin{\Big(2\pi\dfrac{x}{l}\Big)} & x\in[0,L] \\
                  0 & x>L 
        \end{cases},
    \label{eqn:soln-ex1}
    }
           
    We will now consider the homogenized problem. 
    The choice of the unit cell is not unique. 
    We consider two choices: the unit cell starting at $x = 0$ and at a general location $x = a$; for both of these cases, the corner map is chosen to be the left edge of the 1-d unit cell.
    
    First, we set up the lattice $l\scr{L}$ about $x=0$. Thus, there is no partial unit cell on the left, while the partial unit cell on the right is given by $[l\left[\dfrac{L}{l}\right], L]$.
    For this choice, the average polarization in each unit cell and the resulting bound charge density are given by:
    \begin{align}
        & \ve{p}(R)=\dfrac{1}{l}\int_{[R,R+l]}(x-R)\rho(x)\dm x=-\dfrac{1}{2\pi}\cos{\Big(2\pi\dfrac{R}{l}\Big)}=-\dfrac{1}{2\pi}
        \\
        & \rho_{b}(R)=-\divergence_R \ve{p}(R)=0
    \end{align}
    where we have used the fact that the lattice has been set up so that $R\in l\Scr{Z}$.
    
    The approximate homogenized equation\footnote{
    Note that $\varphi^0$ corresponds to the potential on the left boundary and is written as a superscript to differentiate it from the homogenized potential $\varphi_0(x)$.
    } on $[0,L]$ can be written as
    \bml{\varphi_{0}''(R)&=0 \text{ in } \Scr{R}^3\\
         \varphi_0(0)&=\varphi^0\\
         \varphi_0(L)&=\varphi^L.
    }
    The solution is given by:
    \beol{
    \varphi_0(R)=
    \begin{cases} 
              0 & x< 0 \\
              \varphi^0\left(1-\dfrac{x}{L}\right)+\varphi^L\left(\dfrac{x}{L}\right) & x\in[0,L] \\
              0 & x>L
           \end{cases}.
    }
    It is easy to see that $\varphi_0$ above is the point-wise limit of $\varphi_l$ from \eqref{eqn:soln-ex1}, giving the desired weak convergence result.
    
    We next set up the lattice $l\scr{L}$ at $x=a$. 
    Thus there is a partial unit cell $[0,a-l\left[\frac{a}{l}\right]]$ on the left, and $[a+l\left[\frac{L-a}{l}\right],L]$ on the right.
    For this choice, the average polarization in each unit cell and the resulting bound charge density are given by:
    \begin{align}
        & \ve{p}(R)=\dfrac{1}{l}\int_{[R,R+l]}(x-R)\rho(x)\dm x=-\dfrac{1}{2\pi}\cos{\Big(2\pi\dfrac{R}{l}\Big)}=-\dfrac{1}{2\pi}\cos{\Big(2\pi\dfrac{a}{l}\Big)},\label{3.10}
        \\
        & \rho_{b}(R)=-\divergence_R \ve{p}(R)=0
    \end{align}
    The approximate homogenized equation on $[0,L]$ is 
    \bml{
        \varphi_0''(R)&=0 \text{ in } \Scr{R}\\
         \varphi_0(0)&=\varphi^0\\
         \varphi_0(L)&=\varphi^L,
    }
    which is identical to the result from the lattice centered at $x=0$.
\end{example}

We next consider mixed boundary conditions, which show the roles of $\Omega(\Box)$, $\Omega(\lr)$, and $\partial\Omega$. 
We see that the charge distribution has a contribution to the potential by augmenting the flux term on the boundary.
The details of the theorem are available in [\cite{PolarizationTheorem}].
Given a domain $\Omega\subset\Scr{R}^3$ with boundary decomposed as $\partial\Omega=\Gamma_d\cup\Gamma_n$, and given a sequence of ``bulk charge neutral'' charge distribution $\tilde{\rho}_l/l\colon\Omega\to\Scr{R}$ with $||\tilde{\rho}_l||_{L^2(\Box)}<\infty$. 
Then the sequence of electrostatic Poisson problems:
\begin{equation}
\begin{split}
-\Delta &\Phi_l(\ve{x})=\dfrac{\tilde{\rho}_l(\ve{x})}{l} \text{ in } \Omega,\\
	 &\Phi_l=f \text{ on } \Gamma_d,\\
	 &\dfrac{\partial \Phi_l}{\partial \bfn}=g \text{ on } \Gamma_n,\label{Prob1}
\end{split}
\end{equation}
has a well-defined limit as we let $l\to0$ in \eqref{Prob1}. 
We obtain the following homogenized problem:
\bml{
    \Delta_{\ve{x}}&\Phi_0(\ve{x})=\divergence\ve{p}_0 \text{ in } \Omega,\\
	 &\Phi_0=f \text{ on } \Gamma_d,\\
	 &\dfrac{\partial\Phi_0}{\partial \bfn}=g+\ve{p}_0\cdot \ve{n}+\sigma_p \text{ on } \partial\Omega\backslash\Gamma_d\label{Result1}.
}

\begin{example}\label{Ex-2.3}[Mixed Boundary Conditions]
    Consider the charge distribution $\rho\colon\Scr{R}\to\Scr{R}$ with the expression $\rho(x)=\frac{1}{l}\sin{\Big(2\pi\frac{x}{l}\Big)}$ compactly supported on a domain $[0,L]$ (Fig. \ref{Fig.PolProc}). 
    The electric potential $\varphi_l(x)$ is the solution of:
    \bml{
        &\varphi_l''(x)=-\rho(x)~\text{in}~[0,L],\\
         &\varphi_l(0)=f, \\
         &\varphi'_l(L)=g.
    }
    and has the expression:
    \beol{
        \varphi_l(x)=f+gx-\dfrac{x}{2\pi}\cos{\left(2\pi\dfrac{L}{l}\right)}+\dfrac{l}{4\pi^2}\sin{\left(2\pi\dfrac{x}{l}\right)}.
    }
    
    We will now consider the homogenized problem. 
    We present two choices for the unit cell, $x = 0$ and a general $x = a$; for both of these cases, the corner map is chosen to be the left edge of the 1D unit cell.
    
    First, we set up the lattice $l\scr{L}$ about $x=0$. Thus there is no partial unit cell on the left, while the partial unit cell on the right is given by $[l\left[\dfrac{L}{l}\right],L]$.
    The average polarization in each unit cell, the bound charge, the charges due to partial unit cells, and the bound surface charge are given respectively by:
    \begin{align}
        & \ve{p}(R)=\dfrac{1}{l}\int_{[R,R+l]}(x-R)\rho(x)\dm x=-\dfrac{1}{2\pi}\cos{\left(2\pi\dfrac{R}{l}\right)}=-\dfrac{1}{2\pi},
        \\
        & \rho_b(R)=-\divergence\ve{p}(R)=0
        \\
        & \sigma_p(0)=0,~~\sigma_p(L)=\int_{[l\big[\frac{L}{l}\big],L]}\rho(x)\dm x=\dfrac{1}{2\pi}\left[1-\cos{\left(2\pi\dfrac{L}{l}\right)}\right]
        \\
        & \sigma_b(0)=\ve{p}\cdot\ve{n}=-\ve{p}(0)=\dfrac{1}{2\pi},~~\sigma_b\left(l\bigg[\dfrac{L}{l}\bigg]\right)=\ve{p}\left(l\bigg[\dfrac{L}{l}\bigg]\right)=-\dfrac{1}{2\pi}
    \end{align}
    where we have used the fact that the lattice has been set up so that $R\in l\Scr{Z}$.
    
    The approximate homogenized equation on $[0,l\left[\frac{L}{l}\right]]$ can be written as
    \bml{
        \varphi_0''(R)&=0~\text{in}~\Scr{R}\\
         \varphi_0(0)&=f\\
         \left\llbracket\dfrac{\partial\varphi_0}{\partial n}\right\rrbracket\left(l\left[\frac{L}{l}\right]\right)&=\left\llbracket\varphi'_0\left(l\left[\frac{L}{l}\right]\right)\right\rrbracket=\sigma_b\left(l\left[\frac{L}{l}\right]\right)=-\dfrac{1}{2\pi}\\
         \dfrac{\partial\varphi_0}{\partial n}(L)&=\varphi'_0(L)=g+\sigma_p(L)=g+\dfrac{1}{2\pi}-\dfrac{1}{2\pi}\cos{\Big(2\pi\dfrac{L}{l}\Big)}.
    }
         
    We denote by $\varphi_{0,d}$ the potential inside $\Omega(\Box):=[0,l\left[\dfrac{L}{l}\right]]$, and by $\varphi_{0,g}$ the potential in the gap $\Omega(\lr):=[l\left[\dfrac{L}{l}\right],L]$. 
    The continuity conditions $\varphi_{0,d}=\varphi_{0,g}$ and $(\varphi'_{0,d}-\varphi'_{0,g})(l\left[\dfrac{L}{l}\right])=\sigma_b$ are imposed at $l\left[\dfrac{L}{l}\right]$. 
    The solution is:
    \begin{equation}
        \varphi_0(x)
        = 
        \begin{cases}
            \varphi_{0,d}(x)=gx-\dfrac{x}{2\pi}\cos\left(2\pi\dfrac{L}{l}\right)+f, & x \in \Omega(\Box)\\
            \varphi_{0,g}(x)=gx+\dfrac{x}{2\pi}\left[1-\cos\left(2\pi\dfrac{L}{l}\right)\right]-\dfrac{1}{2\pi}l\left[\dfrac{L}{l}\right]+f, & x \in \Omega(\lr)
        \end{cases}
    \end{equation}
    It is clear that for $l$ small, $\varphi_l\approx\varphi_0$. 
    In the limit, $\varphi_{0,g}$ has vanishing support and $\varphi_{0,d}'(l\left[\dfrac{L}{l}\right])\to g+\sigma_p+\sigma_b$.
    
    We next consider the lattice $l\scr{L}$ at $x=a$. 
    Thus there is a partial unit cell $[0,a-l\left[\frac{a}{l}\right]]$ on the left, and $[a+l\left[\frac{L-a}{l}\right],L]$ on the right.
    The average polarization in each unit cell, the bound charge, the charges due to partial unit cells, and the bound surface charge are given respectively by:
    \begin{align}
        & 
        \ve{p}(R)=\dfrac{1}{l}\int_{[R,R+l]}(x-R)\rho(x)\dm x=-\dfrac{1}{2\pi}\cos{\Big(2\pi\dfrac{R}{l}\Big)}=-\dfrac{1}{2\pi}\cos{\Big(2\pi\dfrac{a}{l}\Big)},\label{3.23}
        \\
        & 
        \rho_b(R)=-\divergence_R\ve{p}(R)=0
        \\
        &
        \sigma_p(0)=\int_{[0,a-l\big[\frac{a}{l}\big]]}\rho(x)\dm x=\dfrac{1}{2\pi}\Big[1-\cos{\Big(2\pi\dfrac{a}{l}\Big)}\Big],\\
        &
        \sigma_p(L)=\int_{[a+l\big[\frac{L-a}{l}\big],L]}\rho(x)\dm x=\dfrac{1}{2\pi}\Big[\cos{\Big(2\pi\dfrac{a}{l}\Big)}-\cos{\Big(2\pi\dfrac{L}{l}\Big)}\Big]
        \\
        &
        \sigma_b(a-l\left[\frac{a}{l}\right])=\ve{p}\cdot\ve{n}=-\ve{p}(a-l\left[\frac{a}{l}\right])=\dfrac{1}{2\pi}\cos{\Big(2\pi\dfrac{a}{l}\Big)},\\
        &
        \sigma_b(l\left[\frac{L-a}{l}\right]+a)=\ve{p}(l\left[\dfrac{L-a}{l}\right]+a)=-\dfrac{1}{2\pi}\cos{\Big(2\pi\dfrac{a}{l}\Big)}
    \end{align}
    where we have used that the lattice satisfies $R\in a+l\Scr{Z}$.

    The homogenized problem on $[0,L]$ can be written as:
    \bml{\varphi_0''(R)&=0~\text{in}~\Scr{R}\\
         \varphi_0(0)&=f\\
    \left\llbracket\dfrac{\partial\varphi_0}{\partial \bfn}\right\rrbracket\left(a-l\bigg[\frac{a}{l}\bigg]\right)&=\sigma_b\left(a-l\bigg[\frac{a}{l}\bigg]\right)\\
    \left\llbracket\dfrac{\partial\varphi_0}{\partial \bfn}\right\rrbracket\left(a+l\bigg[\frac{L-a}{l}\bigg]\right)&=\sigma_b\left(a+l\bigg[\frac{L-a}{l}\bigg]\right)\\
    \dfrac{\partial\varphi_0}{\partial \bfn}(L)&=g+\sigma_p(L)}
    We denote by $\varphi_{0,d}$, $\varphi^L_{0,g}$, and $\varphi^R_{0,g}$ respectively the potential inside $\Omega(\Box):=[a-l\bigg[\frac{a}{l}\bigg],a+l\bigg[\dfrac{L-a}{l}\bigg]]$, in the gap to the left $[0,a-l\bigg[\frac{a}{l}\bigg]]$, and in the gap to the right $[a+l\bigg[\dfrac{L-a}{l}\bigg],L]$.
    We require continuity of the potential ($\varphi_{0,d}=\varphi_{0,g}$) at $a+l\bigg[\dfrac{L-a}{l}\bigg]$ and $a-l\bigg[\dfrac{a}{l}\bigg]$.
    Further, the jump in the field must satisfy $(\varphi'_{0,d}-\varphi^{R\prime}_{0,g})(a+l\bigg[\dfrac{L-a}{l}\bigg])=\sigma_b$ and $(\varphi^{L\prime}_{0,g}-\varphi'_{0,d})(a-l\bigg[\dfrac{a}{l}\bigg])=\sigma_b$ respectively.
    The solution is then given by:
    \begin{equation}
        \varphi_0(x)
        = 
        \begin{cases}
            \varphi^L_{0,g}(x)=gx+\dfrac{x}{2\pi}\bigg[\cos\Big(2\pi\dfrac{a}{l}\Big)-\cos\Big(2\pi\dfrac{L}{l}\Big)\bigg]+f
            &
            x \in [0,a-l\bigg[\frac{a}{l}\bigg]]
            \\
            \varphi_{0,d}(x)=gx-\dfrac{x}{2\pi}\cos\Big(2\pi\dfrac{L}{l}\Big)+\Big(a-l\bigg[\dfrac{a}{l}\bigg]\Big)\dfrac{1}{2\pi}\cos\Big(2\pi\dfrac{a}{l}\Big)+f
            &
            x \in \Omega(\Box)
            \\
            \varphi^R_{0,g}(x)=gx+\dfrac{x}{2\pi}\bigg[\cos\Big(2\pi\dfrac{a}{l}\Big)-\cos\Big(2\pi\dfrac{L}{l}\Big)\bigg]+f-\big(l\bigg[\dfrac{a}{l}\bigg]+l\bigg[\dfrac{L-a}{l}\bigg]\big)\dfrac{1}{2\pi}\cos\Big(2\pi\dfrac{a}{l}\Big)
            &
            x \in [a+l\bigg[\dfrac{L-a}{l}\bigg],L]
        \end{cases}
    \end{equation}
    It is clear that for $l$ small enough, $\varphi_l\approx\varphi_0$.
\end{example}

The standard approach to calculating polarization is to deal directly with infinite systems, unlike the limit analysis for finite systems presented above. 
The problem with directly using an infinite system is that there is no boundary, making $\sigma$ undefined. 
This is a critical problem as the surface charge is the quantity corresponding to the polarization that helps preserve the uniqueness of the electric field.

\subsection{Polarization in Infinite Crystal}
\label{PolInfCrys}

We next show that the polarization distribution for an infinite crystal is well-defined. 
Since the transport definition of polarization is defined for infinite crystals, we need to put the two definitions on an equal footing to compare them. 

Given a ``bulk charge neutral'' charge distribution $\tilde{\rho}_l/l\colon\Scr{R}^3\to\Scr{R}$, we consider a sequence of bounded domains $\{\Omega_h\}_{\scr{J}}\subset\Scr{R}^3$ where $\scr{J}$ is some indexing set such that the set increases to fill $\Scr{R}^3$ ($\Omega_h\uparrow\Scr{R}^3$). We require that $||\tilde{\rho}_l||_{L^2(\Box)}<\infty$ with $\tilde{\rho}_l(\ve{x})\to0$ as $|\ve{x}|\to\infty$. The last assumption follows from the requirement that the source term cannot be non-zero at infinity for the potential to be well-defined.

We consider the problem:
\begin{equation}
\begin{split}
    -\Delta &\Phi_{l,h}(\ve{x})=\dfrac{\tilde{\rho}_l(\ve{x})}{l}\mathbbm{1}_{\Omega_h} \text{ in } \Scr{R}^3\backslash\Gamma_{d,h},\\
    & \Phi_{l,h}=f \text{ on } \Gamma_{d,h},\\
    &\Phi_{l,h}=0 \text{ as } |\ve{x}|\to\infty,
\end{split}\label{3.29}
\end{equation}
where $\Gamma_{d,h}\subset\partial\Omega_h$ has Dirichlet boundary conditions specified, and the potential has decay conditions at infinity. 
The homogenized limit of \eqref{3.29} is
\bml{
    \Delta_{\ve{x}}&\Phi_{0,h}(\ve{x})=\mathbbm{1}_{\Omega_h}\divergence\ve{p}_{0,h}~~\text{in}~\Scr{R}^3,\\
	 &\Phi_{0,h}=f~\text{on}~\Gamma_{d,h},\\
	 &\Bigg\llbracket\dfrac{\partial\Phi_{0,h}}{\partial \bfn}\Bigg\rrbracket=\ve{p}_{0,h}\cdot \ve{n}_h+\sigma_{p,h}~\text{on}~\partial\Omega_h\backslash\Gamma_{d,h}.\label{3.30}
}

In \eqref{3.30}, for each $\Omega_h$ bounded, even though $\ve{p}_{0,h}$ is non-unique, $\divergence\ve{p}_{0,h}$ is unique. 
We now let $\Omega_h\uparrow\Scr{R}^3$ and obtain that, for the infinite system, $\divergence\ve{p}_{cl}=\lim_{h\to\infty}\mathbbm{1}_{\Omega_h}\divergence\ve{p}_{0,h}$. Moreover, due to the decay of $\tilde{\rho}_l$, we have $\ve{p}_{0,h}\cdot\ve{n}_h+\sigma_{p,h}\overset{h\uparrow\infty}{\longrightarrow}0$ point-wise. Thus $\Phi_{cl}=\lim_{h\to\infty}\Phi_{0,h}$, is unique.

There is a debate as to whether polarization is a bulk or surface effect. 
One of the main ideas that led to this debate was that for shorted piezoelectric crystals, the change in polarization appears as a surface charge. 
Though some researchers believe that polarization is a bulk effect and not a surface effect, J.W.F. Woo and R. Landauer have claimed that a complete description of polarization must include surface effects [\cite{landauer1981pyroelectricity}]. 
Our derivation reveals that a complete description of the system requires both $\ve{p}_{cl}$ and $\sigma_p$, where the former is the classical definition of the polarization and the latter is the surface charge density previously defined. 
The uniqueness of the potential for finite systems requires both quantities, implying that polarization is not entirely a bulk quantity but requires nonlocal surface-level descriptions.

\section{Transport Definition of Polarization}
\label{Berry}

The Modern Theory of Polarization by Resta and Vanderbilt [\cite{resta2007theory}] originates in the fact that for crystals, a change in polarization produces a measurable current. 
The current, unlike the polarization, is uniquely defined, thereby providing a better definition for polarization. 
The authors highlight that, to date, experimental approaches to measure the polarization have in fact measured the change in polarization, thereby leading them to the hypothesis that the change in polarization is an observable. 
If the process inducing the change is further assumed to be adiabatic, this allows for a perturbative approach to calculating the change in the wavefunction and, thereby, the current. 
This leads to an expression for the current in terms of the Berry phase of the time-dependent wavefunction. 
The Berry phase of the wavefunction gives us the electronic contribution to the polarization; the nuclear contribution is calculated classically; and the net polarization is the sum. 

As stated, the two definitions of polarization $\ve{p}_{cl}$ and $\ve{p}_{T}$ seem disconnected. 
One of the definitions is inherently periodic (at least locally) while the other is not; this was the origin of the non-uniqueness as mentioned in Section \ref{cl}. 
Given a charge distribution $\rho$, the classical definition of polarization is given by
\beol{
    \ve{p}_{cl}(\ve{R},t)=\dfrac{1}{|\Omega_{uc}|}\int_{\ve{R}\oplus\Omega_{uc}}(\ve{x}-\ve{R})\rho(\ve{x},t)\dm V_{\ve{x}},
}
where $\Omega_{uc}$ is the unit cell over which the average is computed.

The corresponding transport definition\footnote{
    As is usual in the adiabatic approximation, we use $t$ to denote the independent variable; $t$ could correspond to time, but it could also be thought of as independently-varied parameter.
} is given by
\beol{
    \Delta\ve{p}_{T}=\ve{J}\Delta t,
}
where
\beol{
    \ve{J}(\ve{R},t)=\dfrac{1}{|\Omega_{uc}|}\int_{\ve{R}\oplus\Omega_{uc}}\ve{j}(\ve{x},t)\dm V_{\ve{x}}
}
is the adiabatically-induced macroscopic current density expressed as a unit cell ($\ve{R}\oplus\Omega_{uc}$) average of the microscopic current density, and $\ve{j}$ is the microscopic current density. 
The important point is that the average current density is unique and does not depend on the cell boundary.

Using the conservation of charge equation, $\partial_t\rho+\divergence\ve{j}=0$, we can show that:
\bml{
    \Delta\ve{p}_{cl}(\ve{R})
    &=\int_{t_i}^{t_f}\int_{\ve{R}\oplus\Omega_{uc}}\dfrac{(\ve{x}-\ve{R})}{|\Omega_{uc}|}\partial_t\rho(\ve{x},t)\dm V_{\ve{x}}\dm t
    =-\int_{t_i}^{t_f}\int_{\ve{R}\oplus\Omega_{uc}}\dfrac{(\ve{x}-\ve{R})}{|\Omega_{uc}|} \divergence(\ve{j}) \dm V_{\ve{x}}\dm t
    \\
    &=-\int_{t_i}^{t_f}\int_{\ve{R}\oplus\Omega_{uc}}\dfrac{[\divergence\big((\ve{x}-\ve{R})\otimes\ve{j}\big)-\ve{j}]}{|\Omega_{uc}|}\dm V_{\ve{x}}\dm t
    \\
    &=\int_{t_i}^{t_f}\ve{J}(\ve{R},t)\dm t-\int_{t_i}^{t_f}\dfrac{1}{|\Omega_{uc}|}\int_{\ve{R}\oplus\partial\Omega_{uc}}(\ve{x}-\ve{R})\ve{n}\cdot\ve{j}\dm S_{\ve{x}}\dm t.
}
We can rewrite the equation above to highlight the difference between the definitions of polarization:
\beol{
    \Delta\ve{p}_{T}(\ve{R})
    =
    \Delta\ve{p}_{cl}(\ve{R})+\int_{t_i}^{t_f}\dfrac{1}{|\Omega_{uc}|}\int_{\ve{R}\oplus\partial\Omega_{uc}}(\ve{x}-\ve{R})\ve{j}\cdot\ve{n} \dm S_{\ve{x}}\dm t
    \label{relation},
}
$\Delta\ve{p}_{T}$ is independent of the choice of unit cell, while both terms on the right are not. 

It is also clear that unless the second term on the right in \eqref{relation} is identically $0$, the two definitions of polarization are not equal. 
We therefore work on rewriting this term next.
While \eqref{relation} is valid for arbitrary shapes of unit cells, we start by choosing unit cells of the form:
\beol{
    \Omega_{uc}:=\{\ve{x}\in\Omega\cap\Scr{R}^3\vert \ve{x}=\nu^i\ve{e}_i, \ve{\nu}\in[0,1)^3 \text{ with basis vectors } \{\ve{e}_i\}_1^3\}.
}
Substituting this into the second term of \eqref{relation}, we get
\beol{
    \dfrac{1}{|\Omega_{uc}|}\int_{\ve{R}\oplus\partial\Omega_{uc}}(\ve{x}-\ve{R})\ve{j}(\ve{x})\cdot\ve{n} \dm S_{\ve{x}}=I_1+I_2+I_3,
}
where $I_k$ is the integral over surfaces with normal co-linear with $\ve{e}^k$. We can expand the lattice vector $\ve{R}=R^i\ve{e}_i$ to simplify the expression for $I_1$:
\bml{
    I_1=&\dfrac{1}{[\ve{e}_1,\ve{e}_2,\ve{e}_3]}\int_{\nu^2=0}^1\int_{\nu^3=0}^1(1.\ve{e}_1+\nu^2\ve{e}_2+\nu^3\ve{e}_3)\ve{j}(R^1+1,\nu^2,\nu^3)\cdot\dfrac{(\ve{e}_2\times\ve{e}_3)}{|\ve{e}_2\times\ve{e}_3|}|\ve{e}_2\times\ve{e}_3|\dm \nu^2\dm \nu^3
    \\
    & \quad -\dfrac{1}{[\ve{e}_1,\ve{e}_2,\ve{e}_3]}\int_{\nu^2=0}^1\int_{\nu^3=0}^1(0.\ve{e}_1+\nu^2\ve{e}_2+\nu^3\ve{e}_3)\ve{j}(R^1,\nu^2,\nu^3)\cdot\dfrac{(\ve{e}_2\times\ve{e}_3)}{|\ve{e}_2\times\ve{e}_3|}|\ve{e}_2\times\ve{e}_3|\dm \nu^2\dm \nu^3
    \\
    = & \dfrac{\ve{e}_1}{[\ve{e}_1,\ve{e}_2,\ve{e}_3]}\int_{\nu^2=0}^1\int_{\nu^3=0}^1\ve{j}(R^1,\nu^2,\nu^3)\cdot(\ve{e}_2\times\ve{e}_3)\dm \nu^2\dm \nu^3
    \\
    = &(\ve{e}_1\otimes\ve{e}^1)\int_{\nu^2=0}^1\int_{\nu^3=0}^1\ve{j}(R^1,\nu^2,\nu^3)\dm \nu^2\dm \nu^3,
}
where we have used the periodicity of $\ve{j}$ and that $[\ve{e}_1,\ve{e}_2,\ve{e}_3]$ is the volume of the unit cell. For the final step, we have used the relation that for 3-d unit vectors, we have $\ve{e}^1=\dfrac{\ve{e}_2\times\ve{e}_3}{[\ve{e}_1,\ve{e}_2,\ve{e}_3]}$.
Similarly,
\bml{
    I_2=(\ve{e}_2\otimes\ve{e}^2)\int_{\nu^1=0}^1\int_{\nu^3=0}^1\ve{j}(\nu^1,R^2,\nu^3)\dm \nu^1\dm \nu^3,
    \quad
    I_3=(\ve{e}_3\otimes\ve{e}^3)\int_{\nu^1=0}^1\int_{\nu^2=0}^1\ve{j}(\nu^1,\nu^2,R^3)d\nu^1d\nu^2.
}

We can then rewrite \eqref{relation} as:
\beol{
    \Delta\ve{p}_{T}(\ve{R})=\Delta\ve{p}_{cl}(\ve{R})+\sum_{k=1}^3\int_{t_i}^{t_f}\int_{\nu\in A_k}(\ve{e}_k\otimes\ve{e}^k)\ve{j}_{(k)}\dm S_{\nu}\dm t,
\label{final}
}
where $\ve{j}_{(k)}$ is short hand for $\ve{j}(\ve{\nu})|_{\nu^k=R^k}$ and $A_k$ is the surface where $\nu_k=0$. 
The key point is that the boundary terms above do not generally appear to simplify to $0$.

We present two counter-examples to showcase this, showing different ways in which the microscopic charge can change locally. 
The charge density amplitude could change, or the charge could move across unit cell boundaries, always maintaining charge neutrality. 
In general, it will be a combination of these two mechanisms.
We analyze them separately for clarity: Example \ref{Ex-4.1} considers the amplitude changing while Example \ref{ex2} considers the charge moving to neighbouring unit cells.

\begin{example}\label{Ex-4.1}
    Consider a flowing periodic charge distribution. 
    This gives rise to a current $\ve{j}\colon[0,1)^3\to\Scr{R}^3$, which we consider of the form $\ve{j}(\ve{x},t)=(\dot{a},0,0)\cos(2\pi n_1x_1)$, where $n_1$ is an integer.
    
    Using the continuity equation, we can integrate in time to get an expression for the change in charge density,
    \bmlnn{
        &\ve{j}(\ve{x},t)=(\dot{a},0,0)\cos(2\pi n_1x_1)
        \\
        \implies & -\dot{\rho} = \divergence\ve{j}=-\dot{a}2\pi n_1\sin(2\pi n_1x_1)
        \\
        \implies&\rho(\ve{x},t)-\rho(\ve{x},0)=[a(t)-a(0)]2\pi n_1\sin(2\pi n_1x_1)
    }
    We set  $\rho(\ve{x},0)$ to have  $0$ mean so that the charge distribution is charge-neutral.
    
    The two definitions of polarization in \eqref{final} have the expressions:
    \begin{align}
        \Delta\ve{p}_{T} & = \int_{t_i}^{t_f}\int_{[0,1]^3}\ve{j}\dm x \dm t=\int_{t_i}^{t_f}\int_{[0,1]^3}(\dot{a},0,0)\cos(2\pi n_1x_1)\dm x \dm t=0
        \\
        \Delta\ve{p}_{cl} & = \int_{t_i}^{t_f}\int_{[0,1]^3}\ve{x}\dot{\rho}\dm x \dm t=\int_{t_i}^{t_f}\dot{a}2\pi n_1\dm t\int_{[0,1]^3}\ve{x}\sin(2\pi n_1x_1)\dm x=-\big(a(t_f)-a(t_i)\big)(1,0,0).
    \end{align}
\end{example}

\begin{example}\label{ex2}
    Consider a charge distribution $\rho\colon\Scr{R}\times\Scr{R}_{>0} \to \Scr{R}$ given by $\rho(x,t)=\sin{(2\pi x-ct)}$; this represents a 1-periodic charge flowing at a speed of ${c}/{2\pi}$. We do not specify the choice of the lattice; this translates to no restrictions on the value of $R$ used below.
    
    The current density, macroscopic current, classical polarization, and the polarization rate have the expressions:
    \begin{align}
        \ve{j}(x,t) &= \ve{v}\rho(x,t)=\dfrac{c}{2\pi}\rho(x,t)=\dfrac{c}{2\pi}\sin{(2\pi x-ct)}
        \\
        \ve{J}(R,t) &= \int_{[R,R+1]}\ve{j}(x,t)\dm x=0
        \\
        \ve{p}(R,t) &= \int_{[R,R+1]}(x-R)\rho(x,t)\dm x=-\dfrac{1}{2\pi}\cos{(2\pi R-ct)} \label{4.15}
        \\
        \dot{\ve{p}}(R,t) &= -\dfrac{c}{2\pi}\sin{(2\pi R-ct)}
    \end{align}
    It is clear that $\dot{\ve{p}}(x,t)\neq\ve{J}(x,t)$.
    The polarization changes with time, but the current does not indicate this physics.
    The difference is reflected in the the boundary term in \eqref{relation} which evaluates to $-\ve{j}(R,t)$.

    To see the lattice-dependence of the classical polarization but the lattice-independence of the macroscopic current, consider a unit cell centered at $x=a<1 \implies R\in \Scr{Z}+a$.
    The classical polarization and current have the expressions:
    \begin{align}
        \ve{p}_a(R,t) &= -\dfrac{1}{2\pi}\cos{(2\pi a-ct)}\implies\dot{\ve{p}}_a(R,t)=-\dfrac{c}{2\pi}\sin{(2\pi a-ct)}
        \\
        \ve{J}_a(R,t) &= 0
    \end{align}
    where the subscript denotes the choice of lattice. 
    Notice that $\dot{\ve{p}}_a(R,t)$ depends on the choice of unit cell through $a$, but $\ve{J}_a(R,t)$ does not.
\end{example}

The fact that the two definitions are unequal does not immediately imply that the transport definition is incorrect. 
One could argue that by getting rid of the part of the polarization that causes the non-uniqueness, we may end up with a better definition of polarization. 
It has been shown in [\cite{thouless1983quantization}] that the change in polarization due to charge crossing the boundary of the unit cell is an integral multiple of the polarization quantum $({e\ve{R}}/{|\Omega_{uc}|})$, where $e$ is the charge of an electron and $\ve{R}$ is a lattice vector; this is why they define the change in polarization modulo this quantum. 
However, it is important to highlight noting that [\cite{thouless1983quantization}] uses the independent electron assumption, which significantly restricts its generality; e.g., it is not valid for strongly correlated electronic systems. 
We will show in this section that $\ve{p}_{cl}$ and $\ve{p}_{T}$ are related in a fundamental way without any assumptions on the physical model.

Though the transport definition of polarization was originally posed for infinite crystals, [\cite{martin2004electronic, ortiz1994macroscopic}] have attempted to extend it to finite systems. 
Their approach assumes a unique decomposition of the total charge in the domain into a surface and a bulk part. 
The criterion for defining the bulk of a domain $\Omega$, denoted $\Scr{B}(\Omega)$, is to set it equal to the region with zero macroscopic electric field:
\beol{
    \Scr{B}(\Omega):=\{\ve{x}\in\Omega\vert\ve{E}_{mac}(\ve{x})=\mathbf{0}\}.
}
The polarization distribution in $\Scr{B}(\Omega)$ is calculated as $\ve{p}_{T}$. 
However, [\cite{ortiz1994macroscopic}] have pointed out that $\ve{E}_{mac}\neq0$ for dielectrics. 
The authors state that there is such a unique decomposition for dielectrics and that the polarization should be calculated based on the charge in this bulk. 
This debate is summarized in [\cite{landauer1981pyroelectricity}]. 

\begin{remark}
    Unlike the method outlined in [\cite{ortiz1994macroscopic}], our approach (Section \ref{cl}) does not require the macroscopic electric field to be zero. As we show there, we can calculate the macroscopic electric potential $\Phi_0$ using \eqref{Result}.
\end{remark}

We now present the main result that will be used to connect the two definitions.

\begin{theorem}\label{Thm4.1}
    Given a sequence of ``bulk charge neutral'' charge distributions $\tilde{\rho}_l/l\colon\Scr{R}^3\times \Scr{R}_{>0}\to\Scr{R}$ with bounded charge in each unit cell ($||\tilde{\rho}_l(\hat{\ve{x}}+l\cdot,t)||_{L^2(\Box)}<\infty$) and decaying far away at infinity ($\rho(\ve{x},t) \to 0 \text{ as } |\ve{x}|\to\infty\text{ and } t\in \Scr{R}_{>0}$). Further, given that the current generated by this charge is bounded ($||\ve{j}_l(\cdot,t)||_{L^2(\Scr{R}^3)}<\infty$).
    
    Then, the bound charge is unique (independent of the choice of unit cell) and is related to the current by the relation:
    \beol{\divergence\dot{\ve{p}}=-\dot{\rho}_b=\divergence\ve{j}}
\end{theorem}

\begin{proof}
    We start by differentiating the electrostatic Poisson equation with respect to $t$:
    \beol{-\Delta\dot{\Phi}=\dot{\rho}=-\divergence\ve{j},\label{potrhocurr}}
    where $\ve{j}$ is the current density and $\Phi$ is the potential. Using $\ve{E}=-\nabla\Phi$, we can rewrite this as:
    \beol{\divergence\ve{j}=-\divergence\dot{\ve{E}}}
    We now substitute $\rho$ and $\ve{j}$, the sequence of charge distribution described above. 
    The corresponding problem is:
    \bml{-\Delta&\dot{\Phi}_l=\dfrac{\dot{\tilde{\rho}}_l}{l}=-\divergence\ve{j}_l~~\text{in}~\Scr{R}^3 \\
    &\dot{\Phi}_l(\ve{x})=0~~\text{as}~|\ve{x}|\to\infty.}
    
    Multiply the preceding equation with $\varphi\in L^2(\Scr{R}^3)\cap\scr{C}_c(\Scr{R}^3)$ and integrate over the whole domain to get:
    \begin{equation}
        -\int_{\Omega}\Delta\dot{\Phi}_l\varphi \dm V_{\ve{x}}=\int_{\Omega}\varphi\dfrac{\dot{\tilde{\rho}}_l}{l}\dm V_{\ve{x}}=-\int_{\Omega}\varphi\divergence\ve{j}_l\dm V_{\ve{x}}
        \implies
        \int_{\Omega}\nabla\dot{\Phi}_l\cdot\nabla\varphi \dm V_{\ve{x}}=\int_{\Omega}\varphi\dfrac{\dot{\tilde{\rho}}_l}{l}\dm V_{\ve{x}}=\int_{\Omega}\nabla\varphi\cdot\ve{j}_l\dm V_{\ve{x}},\label{premise}
    \end{equation}
    where $\Omega$ is the support of $\varphi$.
    
    Since each term is bounded, we can extract a unique convergent subsequence. Taking the limit of \eqref{premise} and making use of \eqref{polthm}, we obtain,
    \begin{equation}
        \int_{\Omega}\nabla\dot{\Phi}_0\cdot\nabla\varphi \dm V_{\ve{x}}=\int_{\Omega}\nabla\varphi\cdot\dot{\ve{p}}\dm V_{\ve{x}}=\int_{\Omega}\nabla\varphi\cdot\ve{j}_0\dm V_{\ve{x}}
        \implies
        \int_{\Omega}\Delta\dot{\Phi}_0\varphi \dm V_{\ve{x}}=\int_{\Omega}\varphi\divergence\dot{\ve{p}}\dm V_{\ve{x}}=\int_{\Omega}\varphi\divergence\ve{j}_0\dm V_{\ve{x}}\label{result}.
    \end{equation}
    Using the density of the space, we arrive at the desired strong form result:
    \beol{\Delta\dot{\Phi}_0=\divergence\dot{\ve{p}}=\divergence\ve{j}_0.}
    Since $\Phi_0$ is unique, being a weak limit of $\Phi_l$, we have that $\divergence\dot{\ve{p}}$ and $\divergence\ve{j}_0$ are unique and equal to each other. 
    They are both equal to $-\dot{\rho}_b$, the rate of bulk bound charge.

    Thus, even though $\ve{p}$ is not unique for the infinite system, its divergence is unique and can be used to calculate the potential for an infinite system. Moreover, the divergence of the current will give us the rate of change of the potential for the infinite system; they both give the same potential and, hence, the same electric field and energy. Note that this approach fails for finite systems as the surface charge does not come from polarization.
\end{proof}

\begin{remark}
    Note that \eqref{potrhocurr} can be obtained in the electro-quasi-static limit when the magnetic field can be neglected.
    We assume these equations hold true under the adiabatic hypothesis of the Modern Theory of Polarization.
\end{remark}

We have shown that $\divergence\dot{\ve{p}}=\divergence\ve{j}=-\dot{\rho}_b$, implying that the current and the polarization are related via the bulk bound charge. 
Since the electrostatics for infinite domains depend only on this, $\ve{p}_{T}=\int\ve{j}\dm t$ provides a good description of the polarization for infinite domains. 
Because we take the divergence of the polarization to obtain the bound charge, $\ve{j}$ provides a useful measure of the change in polarization. 
From Theorem \ref{Thm4.1}, we should note that it is not clear what the corresponding boundary term would be for finite domains. 
Thus, we connect $\ve{p}_{cl}$ for the infinite crystal to $\ve{p}_{T}$. 
Towards the end of Section \ref{cl}, we have shown that $\ve{p}_{cl}$ for an infinite domain presents a well-defined electrostatic problem. 

\begin{remark}
	\label{ReMark4.3}
    The adiabatic approximation is essential to the transport definition of polarization.
    However, this restricts its applicability, e.g., displacive ferroelectric phase transitions in which there is an abrupt and essentially discontinuous change in the polarization across the transition [\cite{james1986displacive,devonshire1954theory}].
\end{remark}

\section{Energetic Definition of Polarization}\label{W}

In this section, we will show that $\ve{p}_W$ is the same as $\ve{p}_{cl}$. In order to do so, we establish a relation between the variation of the potential and the variation of the total energy. We treat the electrons quantum-mechanically, while we treat the nuclei classically using the Born-Oppenheimer approximation [\cite{parr1980density}].

\subsection{Variational derivative with respect to applied electric field}\label{VariationalDerivative} 
We start by considering the electronic contribution to the variation in energy. Consider a system with $M$ nuclei and $N/M$ electrons per nucleus. The position of these nuclei are $\{\ve{R}_{\alpha}\}_1^M$. 
We will use $\ve{x}_i\in\Scr{R}^3$ to denote the coordinate of the $i$-th electron. 

We consider the Schr\"{o}dinger equation in the presence of an external potential $\phi\colon\Scr{R}^3\to\Scr{R}$:
\begin{equation}
    \TO{H}\ket{\Psi}=W^{tot}_{el}\ket{\Psi}, \quad \text{ subject to } \int_{\Scr{R}^{3N}}|\Psi|^2\dm V_{\ve{x}_1}\ldots \dm V_{\ve{x}_N}=1
\label{Schro}
\end{equation}
where $W^{tot}_{el}$ is the total energy of the system, $\TO{H}$ is the Hamiltonian, and $\Psi\colon\Scr{R}^{3N}\to\Scr{C}$ is the electron wave-function representing the minimum energy configuration of the electrons. $\TO{H}$ has the form:
\beol{
    \TO{H}
    =
    -\dfrac{\hbar^2}{2m}\sum_{i=1}^N\Delta_i
    +\dfrac{e^2}{4\pi}\left[
        \sum_{i,j=1; i<j}^{N}\dfrac{1}{|\ve{x}_i-\ve{x}_j|}
        -\sum_{\alpha=1}^M\sum_{i=1}^N\dfrac{Z_{\alpha}}{|\ve{R}_{\alpha}-\ve{x}_i|}
    \right]
    -\sum_{i=1}^Ne\phi(\ve{x}_i),
}
where $\Delta_i$ is the Laplacian acting on the $i$-th component of $\Psi$; $e$ is the absolute charge of an electron; $m$ is the mass of the electron; and $Z_{\alpha}$ is the atomic number of the $\alpha$-th nucleus located at $\ve{R}_{\alpha}$.
The electronic charge distribution $\rho_e(\ve{x})$ is related to $\Psi$ by:
\beol{\rho_e(\ve{x})=-eN\int_{\Scr{R}^{3(N-1)}}|\Psi(\ve{x},\ve{x}_1,\ldots,\ve{x}_{N-1})|^2\dm V_{\ve{x}_1}\ldots \dm V_{\ve{x}_{N-1}}.}

We consider a variation of the external potential $\phi(\ve{x})\to\phi(\ve{x})+\delta\phi(\ve{x})$. 
This causes a variation in the wavefunction $\ket{\Psi}+\ket{\delta\Psi}$, the energy $W^{tot}_{el}+\delta W^{tot}_{el}$, and the Hamiltonian  $\TO{H}-e\sum\delta\phi(\ve{x}_i)$. Retaining leading order terms:
\begin{align}
    &
    \TO{H}\ket{\Psi}+\TO{H}\ket{\delta\Psi}-e\Big[\sum_1^N\delta\phi(\ve{x}_i)\Big]\ket{\Psi}=W^{tot}_{el}\ket{\Psi}+W^{tot}_{el}\ket{\delta\Psi}+\delta W^{tot}_{el}\ket{\Psi}\label{5.4}
    \\
    &
    \implies \TO{H}\ket{\delta\Psi}-e\Big[\sum_1^N\delta\phi(\ve{x}_i)\Big]\ket{\Psi}=W^{tot}_{el}\ket{\delta\Psi}+\delta W^{tot}_{el}\ket{\Psi}
    \\
    &
    \implies\bracket{\Psi}{\TO{H}}{\delta\Psi}-\bracket{\Psi}{\Big[e\sum_1^N\delta\phi(\ve{x}_i)\Big]}{\Psi}=\bracket{\Psi}{W^{tot}_{el}}{\delta\Psi}+\bracket{\Psi}{\delta W^{tot}_{el}}{\Psi}
    \\
    &
    \implies -\bracket{\Psi}{\Big[e\sum_1^N\delta\phi(\ve{x}_i)\Big]}{\Psi}=\bracket{\Psi}{\delta W^{tot}_{el}}{\Psi}
\end{align}
where we have used \eqref{Schro} to reach the second step, taken the inner-product with $\bra{\Psi}$ to reach the third step, and used \eqref{Schro} in combination with the inner-product with $\bra{\Psi}$ to reach the final step.
If we explicitly write out the inner-products in the final expression above, we have:
\beol{\delta W^{tot}_{el}=\int_{\Scr{R}^3}\delta\phi(\ve{x})\rho_e(\ve{x})\dm V_{\ve{x}}\label{var-rel_el}.}

We next turn to the contribution of the nuclei to the variation in energy. 
We consider the potential due to the $M$ nuclei and their corresponding electron distribution being represented by a potential energy $V\colon\Scr{R}^{3M}\to\Scr{R}$. 
The force due to this potential on the nucleus at $\ve{R}_{\alpha}$ is $-\nabla_{\alpha}V(\ve{R}_1,\ldots,\ve{R}_M)$.
Further, as earlier, $\phi\colon\Scr{R}^3\to\Scr{R}$ is the external potential, with a corresponding force $-eZ_{\alpha}\nabla\phi(\ve{R}_{\alpha})$.
At equilibrium, we require:
\begin{equation}
    \nabla_{\alpha}V(\tilde{\ve{R}})+eZ_{\alpha}\nabla\phi(\ve{R}_{\alpha})=0~~\forall\alpha\in\{1,\ldots,M\}~,\label{MolStat}
\end{equation}
where we use the notation $\tilde{\ve{R}}$ to denote the set of vectors $\{\ve{R}_1,\ldots,\ve{R}_M\}$.

Now we vary the potential $\phi\to\phi+\delta\phi$. 
This will cause a variation in the nuclear positions $\ve{R}_{\alpha}\to\ve{R}_{\alpha}+\delta\ve{R}_{\alpha}$ to remain in equilibrium. 
A Taylor expansion to leading order gives:
\begin{align}
    &
    V(\tilde{\ve{R}}+\delta\tilde{\ve{R}}) 
    = 
    V(\tilde{\ve{R}})+\sum_{\alpha}\nabla_{\alpha}V(\tilde{\ve{R}})\cdot\delta\ve{R}_{\alpha}
    \\
    &
    (\phi+\delta\phi)(\ve{R}_{\alpha}+\delta\ve{R}_{\alpha})
    =
    \phi(\ve{R}_{\alpha})+\nabla\phi(\ve{R}_{\alpha})\cdot\delta\ve{R}_{\alpha} + \delta\phi(\ve{R}_{\alpha})+\nabla\delta\phi(\ve{R}_{\alpha})\cdot\delta\ve{R}_{\alpha}
\end{align}

We use these expressions to compute the variation in the total energy to leading order:
\begin{equation}
    W^{tot}_{n}
    =\sum_{\alpha}eZ_{\alpha}\phi(\ve{R}_{\alpha})+V(\tilde{\ve{R}})
    \implies
    W^{tot}_n+\delta W^{tot}_n
    = \sum_{\alpha}eZ_{\alpha}(\phi+\delta\phi)(\ve{R}_{\alpha}+\delta\ve{R}_{\alpha}) + V(\tilde{\ve{R}}+\delta\tilde{\ve{R}})
\end{equation}
and then use \eqref{MolStat} to compute:
\begin{align}
    \delta W^{tot}_n
    & = \sum_{\alpha}eZ_{\alpha}\big(\nabla\phi(\ve{R}_{\alpha})\cdot\delta\ve{R}_{\alpha}+\delta\phi(\ve{R}_{\alpha})\big)+\sum_{\alpha}\nabla_{\alpha}V(\tilde{\ve{R}})\delta\ve{R}_{\alpha}
    \\
    & = 
    \sum_{\alpha}\underbrace{\left[eZ_{\alpha}\nabla\phi(\ve{R}_{\alpha})+\nabla_{\alpha}V(\tilde{\ve{R}})\right]}_{=0}\cdot\delta\ve{R}_{\alpha} + \sum_{\alpha}eZ_{\alpha}\delta\phi(\ve{R}_{\alpha})
    \\
    & = 
    \sum_{\alpha}eZ_{\alpha}\delta\phi(\ve{R}_{\alpha}) \label{5.13}
\end{align}
Formally writing the nuclear charge distribution in terms of a density $\rho_n$, we can rewrite \eqref{5.13} as:
\beol{
    \delta W^{tot}_n=\int_{\Scr{R}^3}\rho_n(\ve{R})\delta\phi(\ve{R})\dm V_{\ve{R}},\label{var-rel_n}
}

The variation in the total energy is the sum of the individual expressions from \eqref{var-rel_el} and \eqref{var-rel_n}:
\beol{
    \delta W=\int_{\Scr{R}^3}\rho(\ve{x})\delta\phi(\ve{x})\dm V_{\ve{x}}\label{var-rel},
}
where $\rho$ is the net charge distribution.

\begin{remark}
    The energetic definition of polarization is defined as the derivative of the free energy density with respect to the electric field. 
    While the electric field has contributions from the external potential as well as internal charges (nuclei and electrons), because we are considering a variation of the system at equilibrium, only the external potential explicitly appears in the final expression.
\end{remark}

We are now in a position to homogenize \eqref{var-rel} --- i.e., consider a sequence of nucleii and the corresponding electron distributions --- to get the desired relation between the polarization, surface charge, and the variation in the energy.
The result is analogous to the polarization theorem stated in \eqref{polthm}. 
We consider a finite crystalline domain and homogenize \eqref{var-rel}, appropriately scaled, by following a procedure similar to that used in [\cite{PolarizationTheorem}]. 
The scaling is chosen to get a finite, i.e., neither going to infinity nor to zero, polarization in the limit. 
However, unlike the situation in [\cite{PolarizationTheorem}], the charge density cannot be \textit{a priori} be assumed to be compactly supported or to have zero-mean over a unit cell, we have to allow for this more general setting.
A rigorous proof is presented in Section \ref{Polarization from Schrodinger's equations}, while the final result is presented below.
\beol{
    \delta \tilde{W}^{tot}=\lim_{l\to 0}\int_{\Scr{R}^3}\dfrac{\rho(\ve{x})}{l}\delta\phi(\ve{x})\dm V_{\ve{x}}=\int_{\partial\Omega}\delta\phi(\ve{s})\sigma(\ve{s})\dm S_{\ve{s}}+\int_{\Omega}\delta\ve{E}\cdot\ve{p}\dm V_{\ve{x}}.\label{Appendix-final}
}
Interpreting \eqref{Appendix-final}, we have that the variation in the energy has the polarization as energy-conjugate to the electric field in the bulk, and has the surface charge density as energy-conjugate to the potential on the surface.
Thus, the variation requires contribution from both the surface charge and the bulk polarization. 
This aligns with the finding that $(\ve{p},\sigma)$ comes as a pair; both are required to specify the electrostatic fields and energetic response.
Further, from \eqref{Appendix-final}, we see that $\ve{p}_W=\ve{p}_{cl}$.

\subsection{Polarization from the Schrodinger Equation}\label{Polarization from Schrodinger's equations}

We present a derivation of the relation between polarization and the free energy of the system. 
We will start from \eqref{var-rel} to derive our final result in \eqref{Appendix-final} for zero temperature; note that this derivation is not required for the physical interpretation of \eqref{Appendix-final}.

We first write down some notation that we shall use in the derivation present in this section. In order to show that $\ve{p}_W=\ve{p}_{cl}$ in Section \ref{W}, we require assumptions on the charge density that we discuss first.

\begin{definition}[Augmented domain $\Omega_{h}$]
	Given a domain $\Omega\subset\Scr{R}^3$, we define the augmented region 
	\beol{\Omega_h:=\Omega\cup\{\ve{x}\in\Scr{R}^3\vert\exists \ve{s}\in\partial\Omega~\text{such that}~\ve{x}\in[\ve{s},\ve{s}+h\ve{n}(\ve{s})]\},}
	where $\ve{n}$ is the outward facing normal to $\partial\Omega$.
\end{definition}

\begin{definition}[$1$-Concentrated near $\Omega$]
    A sequence of charge distributions is said to be $1$-Concentrated near $\Omega\subset\Scr{R}^3$ if $\exists d>0$ such that
    \beol{ g_1(t)\leq|\rho_l(\ve{s}+lt\ve{n}(\ve{s}))|\leq g_2(t)~\forall t>d,}
    where $g_1,g_2\in L^1(\Scr{R})$.

     Note that the definition of $1$-concentrated is related to a general definition of $l$-concentrated, where we have to establish bounds above and below. 
     For the present case, one may consider $g_1=0$.
\end{definition}

\begin{remark}
	The idea behind the ``augmented domain'' and ``1-Concentrated near $\Omega$'' is that we want to include a region outside $\Omega$ and have a hypothesis on how the charge behaves in this region. The charge not in $\Omega$ must be concentrated near $\partial\Omega$ to identify it as a surface charge.
\end{remark}

\begin{definition}[good boundary]
	Given a domain $\Omega\subset\Scr{R}^3$. We say that the boundary $\partial\Omega$ is ``good'' iff
	\beol{\#(\Omega(\lr))=C\dfrac{|\partial\Omega|}{l^{2}}+o(1)}
	where $\#$ denote the counting measure and $o(1)$ represents a quantity that goes to 0 as $l\to0$.
\end{definition}

We now present the proof of how \eqref{var-rel} leads to \eqref{Appendix-final}. Our material of interest is ionic crystals; thus, we make some simplifying assumptions that we will have charge-neutral unit cells in the bulk of the material.

Consider a domain $\Omega\subset\Scr{R}^3$.
The location of the nuclei is given by $\scr{L}\cap\ve{q}\oplus\Omega$ for some $\ve{q}\in\Scr{R}^3$. 
We want to analyze this system in the presence of an external field $\phi\colon\Scr{R}^3\to\Scr{R}$; the variation in the energy due to a variation in potential is given by \eqref{var-rel}. 
We aim to replace the system of nuclear and electron distributions with a homogenized problem that is close, in some sense, to the actual system. 
To find the homogenized problem, consider the sequence of problems laid out with a scaled lattice $l\scr{L}$, where $l\in(0,1]$ is a scaling parameter. 
As we scale the lattice ($l\to 0$), we have a sequence of charge distributions denoted by $\rho_l\colon\Scr{R}^3\to\Scr{R}$ with $l\in(0,1]$, $\{\rho_l(\cdot+l\ve{y})\}_{(0,1]}\in L^2(\Box)$ and $\{\rho_l\}_{(0,1]}\vert_{\Omega^c}\in L^{1}$. The latter condition follows from the charge being ``1-concentrated'' near $\Omega$. Moreover, we require $\delta\phi\in W^{1,2}(\Omega)\cap\scr{C}^0(\Omega^c)\cap L^{\infty}(\Omega^c)$. In addition to the assumption on $\rho_l$, we shall assume that the boundary, $\partial\Omega$, is ``good''. The definition of $\Box$ is in Section \ref{Notation}.

In order to homogenize, we need to establish the bulk and surface regions; the bulk will give us a well-defined polarization, while the surface will give us a well-defined surface charge density. The bulk is denoted by $\scr{B}_l(\Omega)$ and is defined as
\beol{\scr{B}_l(\Omega):=\left\{\hat{\ve{x}}\oplus l\Box\vert \hat{\ve{x}}\oplus l\Box\subset\Omega \text{ and } \int_{\hat{\ve{x}}\oplus l\Box}\rho_l(\ve{x})\dm V_{\ve{x}}=0\right\}.}

The transition region is denoted by $\scr{T}_l(\Omega)$ and is defined as
\beol{\scr{T}_l(\Omega):=\left\{\hat{\ve{x}}\oplus l\Box\vert \hat{\ve{x}}\oplus l\Box\subset\Omega \text{ and } \int_{\hat{\ve{x}}\oplus l\Box}\rho_l(\ve{x})\dm V_{\ve{x}}\neq0\right\}.}

Since the charge extends outside $\Omega$, we also need to consider this region. We denote this region as $\scr{S}_l(\Omega)$ and define it as
\beol{\scr{S}_l(\Omega):=\left\{\hat{\ve{x}}\oplus l\Box\vert \hat{\ve{x}}\oplus l\Box\cap\partial\Omega\neq\{\}\right\}.}
where the latter represents partial unit cells.

\begin{remark}
    The region $\scr{T}_l$ may be decomposed as 
    \beol{
        \left\{U_{l,\alpha}\subset\scr{T}_l(\Omega)|\int_{U_{l,\alpha}}\rho_l(\ve{x})\dm V_{\ve{x}}=0,~\bigcup_{\alpha}U_{l,\alpha}=\scr{T}_l(\Omega) \text{ and } \forall\alpha,\beta\in\scr{J}_l~~ U_{l,\alpha}\cap U_{l,\beta}=\{\}~\forall l\in\scr{I}\right\},
    \label{Ula}}
    where $\scr{J}_l$ is an arbitrary indexing set which ennumerates the collection $\{U_{l,\alpha}\}_{\alpha\in\scr{J}_l}$ for each $l\in\scr{I}$. Here $\scr{I}\subseteq (0,1]$ denoting a possible subsequence may be extracted if necessary. It may be possible to extract a well-defined polarization from these sets. However, since it is hard to predict the evolution of $U_{l,\alpha}$ with $l$, we do not know if a certain sequence $\{U_{l,\alpha}\}_{\scr{I}}$ may tend to become degenerate in the limit. This is possible if a set $U_{l,\alpha}$ scales differently in some directions than the unit cells in the bulk. The unit cells in the bulk scale with $l$. For the sake of brevity, we will call such a sequence of sets degenerate. We assume that a reasonable way exists to identify a set $U_{l,\alpha}$ along $l\in\scr{I}$.
	
	The problem with degeneracy is that it becomes hard to establish the convergence of gradients. Visually, one can think of the set becoming long and thin as $l\to0$. Thus, along the long direction, the external potential may not be accurately approximated by the affine term in its Taylor series. The linear term (electric field) couples with the polarization, while the higher-order terms couple with higher-order moments; however, we do not want higher-order moments to describe the electronic nature of our material in the current work.  
	
	The collection of sets that are not degenerate along $l\in\scr{I}$ can be combined  with $\scr{B}_l(\Omega)$ to redefine the bulk ($\scr{B}^m_l(\Omega)$). The degenerate sets will be used to redefine the transition region ($\scr{T}^m_l(\Omega)$). The superscript $m$ denotes modified. Note that the decomposition in \eqref{Ula} is not unique. We impose the condition that the decomposition be chosen such that $|\scr{T}^m_l(\Omega)|$ is the smallest possible. 
\end{remark}

\begin{remark}
    We can combine some (if not all) sets in $T_l$ with $S_l$. Thus, we can redefine $S_l$ by adding these degenerate sequences of sets. The modified $S_l(\Omega)$, denoted $S^m_l(\Omega)$, must contain $S_l(\Omega)$ and sets, which will tend to a degenerate shape along one or more directions. 
		
		Notice that these new sets in the modified surface region consist (by definition) of charge-neutral sets. Thus, when we establish convergence to a limit, we will get a potential times the total charge. This will be clear from the proof later. Since the total charge is zero, there is no contribution in the limit. The sets in $T^m_l(\Omega)$ are assumed to vanish in measure and/or in count, being redistributed to $B^m_l(\Omega)$ and $S^m_l(\Omega)$. The assumption made on $T^m_l(\Omega)$ below is a simpler way to achieve this. This is achieved because, as we noted earlier, the decomposition into $U_{l,\alpha}$ is non-unique, so if we assume a set exists such that the $T^m_l(\Omega)$ is minimal possible and vanishes in the limit, we have considerably simplified the analysis. We take the simpler route since we are interested in the least restrictive assumption. We will drop the superscript $m$ if no confusion arises.
\end{remark}

\begin{assumption}
    We assume the ``modified'' regions in $\Omega$ satisfy $|\scr{S}^m_l(\Omega)|<<|\scr{B}^m_l(\Omega)|~\forall l$ with 
    $$\lim_{l\to0}\dfrac{|\scr{S}^m_l(\Omega)|}{|\scr{B}^m_l(\Omega)|}=0,$$
    where the above expresses the fact that $\scr{S}^m_l$ is surface level quantities. In addition, we assume that
\beol{\lim_{l\to0}|\scr{T}^m_l(\Omega)|=0,}
where the above expresses the fact that $\scr{T}^m_l$ does not retain any set of finite measures in the limit. 
\end{assumption}

In the thermodynamic limit, following [\cite{Sen2022}], we let $\delta W^{tot}$ scale with $l$.
In addition, we make the assumption that $\rho_l$ on $\Omega^c$ is 1-Concentrated near $\partial\Omega$. 

Scaling the variation in the total energy using \eqref{var-rel} gives:
\beol{
    \lim_{l\to0}\dfrac{\delta W^{tot}}{l}=\lim_{l\to0}\int_{\Scr{R}^3}\delta\phi\dfrac{\rho_l}{l}\dm V_{\ve{x}}=\lim_{h\to\infty}\lim_{l\to0}\left[\int_{\Omega_{lh}}\delta\phi\dfrac{\rho_l}{l}\dm V_{\ve{x}}+\int_{\Omega_{lh}^c}\delta\phi\dfrac{\rho_l}{l}\dm V_{\ve{x}}\right],
\label{eqn:proof-1}
}
where we first take the limit of $l\to0$ for a fixed $h$.
We get a well-defined surface charge and polarization on $\Omega$ from the first term, while the second term tends to $0$ as $h\uparrow\infty$.  $\Omega_{lh}$ is an augmented domain that helps prove convergence.

We first work on the second term in \eqref{eqn:proof-1}.
Using that $\delta\phi\vert_{\Omega^c}\in \scr{C}^0\cap L^{\infty}$, we get:
\begin{equation}
\begin{split}
    \left\vert\int_{\Omega_{lh}^c}\delta\phi\dfrac{\rho_l}{l}\dm V_{\ve{x}}\right\vert
    & \leq
    \sup_{\Omega^c}\delta\phi(\ve{x})\int_{\partial\Omega\times(h,\infty)}|\rho_l(\ve{s}+lt\ve{n}(\ve{s}))|~|\eta(\ve{s}+lt\ve{n}(\ve{s}))|\dm t\dm S_{\ve{s}}
    \\
    & \leq C\int_{\partial\Omega\times(h,\infty)}|\rho_l(\ve{s}+lt\ve{n}(\ve{s}))| \dm t\dm S_{\ve{s}} 
    \leq C\int_{\partial\Omega\times(h,\infty)}g_2(t)\dm t\dm S_{\ve{s}}
    = G(h),
\end{split}
\label{eqn:proof-2}
\end{equation}
where $\eta\colon\Scr{R}^3\to\{0,1\}$ is used to prevent over-counting when performing the integration; especially when normals intersect. To reach the last line, we have used the integrability of $g_2$, $\eta\leq1$, and $|\partial\Omega|<\infty$.

We next work on the first term from \eqref{eqn:proof-1}:
\bml{
    \int_{\Omega_{lh}}\delta\phi\dfrac{\rho_l}{l}\dm V_{\ve{x}}
    =
    &
    \sum_{\hat{\ve{x}}\in\scr{B}_l(\Omega)}l^3\int_{\Box}\delta\phi(\hat{\ve{x}}+l\ve{y})\dfrac{\rho_l(\hat{\ve{x}}+l\ve{y})}{l}\dm V_{\ve{y}}
    +
    \sum_{\hat{\ve{x}}\in\scr{S}_l(\Omega)}l^3\int_{\lr}\delta\phi(\hat{\ve{x}}+l\ve{y})\dfrac{\rho_l(\hat{\ve{x}}+l\ve{y})}{l}\dm V_{\ve{y}}
    \\
    & + \int_{\partial\Omega\times[0,h]}\delta\phi(\ve{s}+lt\ve{n}(\ve{s}))\rho_l(\ve{s}+lt\ve{n}(\ve{s}))\eta(\ve{s}+lt\ve{n}(\ve{s}))\dm t\dm S_{\ve{s}},
\label{eqn:proof-3}
}
where the first term, referred to as the bulk term, will give the polarization in the bulk; the subsequent two terms will give a well-defined (inner and outer, respectively) surface charge. 
Here $\ve{n}(\ve{s})$ represents the outward normal and $\eta\colon\Scr{R}^3\to\{0,1\}$ prevents over-counting.

We work next on the bulk term from \eqref{eqn:proof-3}:
\bml{
    \sum_{\hat{\ve{x}}\in\scr{B}_l(\Omega)}l^3\int_{\Box}\delta\phi(\hat{\ve{x}}+l\ve{y})\dfrac{\rho_l(\hat{\ve{x}}+l\ve{y})}{l}\dm V_{\ve{y}}
    =
    &
    \sum_{\hat{\ve{x}}\in\scr{B}_l(\Omega)}l^3\delta\varphi(\hat{\ve{x}})\int_{\Box}\dfrac{\rho_l(\hat{\ve{x}}+l\ve{y})}{l}\dm V_{\ve{y}}
    +
    \sum_{\hat{\ve{x}}\in\scr{B}_l(\Omega)}l^3\delta\nabla\varphi(\hat{\ve{x}})\cdot\int_{\Box}\ve{y}\rho_l(\hat{\ve{x}}+l\ve{y})\dm V_{\ve{y}}
    \\
    &
    + 
    \sum_{\hat{\ve{x}}\in\scr{B}_l(\Omega)}l^3\int_{\Box}\underbrace{\dfrac{\delta\phi(\hat{\ve{x}}+l\ve{y})-\delta\phi(\hat{\ve{x}})-l\ve{y}\cdot\nabla\delta\varphi(\hat{\ve{x}})}{l|\ve{y}|}}_{u_l(\hat{\ve{x}},\ve{y})}|\ve{y}|\rho_l(\hat{\ve{x}}+l\ve{y})\dm V_{\ve{y}}
\label{eqn:proof-4}
}
where we have used that $\nabla$ and $\delta$ commute. 
The first term on the right side of \eqref{eqn:proof-4} is zero from bulk charge neutrality.
The third term on the right side of \eqref{eqn:proof-4} goes to zero in the limit from $\delta\phi$ being absolutely differentiable:
\bml{
    \left\vert\sum_{\hat{\ve{x}}\in\scr{B}_l(\Omega)}l^3\int_{\Box}u_l(\hat{\ve{x}},\ve{y})|\ve{y}|\rho_l(\hat{\ve{x}}+l\ve{y})\dm V_{\ve{y}}\right\vert
    &
    \leq
    \sum_{\hat{\ve{x}}\in\scr{B}_l(\Omega)}l^3|u_l(\hat{\ve{x}},\ve{y})|_{L^2(\Box)}\underbrace{|\ve{y}\rho_l(\hat{\ve{x}}+l\ve{y})|_{L^2(\Box)}}_{\leq C}
    \\
    &
    \leq
    \sup_{\hat{\ve{x}}\in\scr{B}_l(\Omega)}|u_l(\hat{\ve{x}},\ve{y})|_{L^2(\Box)}\sum_{\hat{\ve{x}}\in\scr{B}_l(\Omega)}l^3C
    \leq C\sup_{\hat{\ve{x}}\in\scr{B}_l(\Omega)}|u_l(\hat{\ve{x}},\ve{y})|_{L^2(\Box)}
    \overset{l\downarrow0}{\longrightarrow}0,
 }
where the last line follows from the total differentiability of $\delta\phi$ and the sandwich theorem. 

Taking the limit of the second term on the right side of \eqref{eqn:proof-4} gives us:
\beol{
    \int_{\Omega}\nabla\delta\phi\cdot\ve{p}\dm V_{\ve{x}}, \quad \text{ with } \quad \ve{p}(\ve{x}):=\int_{\Box}\ve{y}\rho_0(\ve{x},\ve{y})\dm V_{\ve{y}}
}
and $\rho_0$ is the 2-scale limit of the sequence $\rho_l$.

We now turn to the second term from \eqref{eqn:proof-3} (the inner surface term): 
\begin{equation}
\begin{split}
    & 
    \sum_{\hat{\ve{x}}\in\scr{S}_l(\Omega)}l^3\int_{\lr}\delta\phi(\hat{\ve{x}}+l\ve{y})\dfrac{\rho_l(\hat{\ve{x}}+l\ve{y})}{l}\dm V_{\ve{y}}
    \\
    &
    =
    \underbrace{\sum_{\hat{\ve{x}}\in\scr{S}_l(\Omega)}l^2\delta\phi(\hat{\ve{x}})\int_{\lr}\rho_l(\hat{\ve{x}}+l\ve{y})\dm V_{\ve{y}}}_{\text{Term I}}+\underbrace{\sum_{\hat{\ve{x}}\in\scr{S}_l(\Omega)}l^2\int_{\lr}\underbrace{\big[\delta\phi(\hat{\ve{x}}+l\ve{y})-\delta\phi(\hat{\ve{x}})\big]}_{u_l(\hat{\ve{x}},\ve{y})}\rho_l(\hat{\ve{x}}+l\ve{y})\dm V_{\ve{y}}}_{\text{Term II}}
    \label{eqn:proof-5}
\end{split}
\end{equation}
Applying the Cauchy-Schwartz inequality on the second term from \eqref{eqn:proof-5}, we obtain:
\beol{\text{Term II}\leq\sum_{\hat{\ve{x}}\in\scr{S}_l(\Omega)}l^2\big|u_l(\hat{\ve{x}},\ve{y})\big|_{L^2(\lr)}|\rho_l|_{L^2(\lr)}\leq C\sup_{\scr{S}_l(\Omega)}\big|u_l(\hat{\ve{x}},\ve{y})\big|_{L^2(\lr)}\overset{l\downarrow0}{\longrightarrow}0,}
where we have used the fact that the boundary is ``good'' in order to derive the second inequality. In the first inequality, we have used the fact that either $|\rho_l|_{L^2(\lr)}\leq 1$ or $>1$. If its the former, then we are done as we have a bound; else, we have $|\rho_l|_{L^2(\lr)}\leq|\rho_l|_{L^2(\lr)}^2\leq |\rho_l|^2_{L^2(\Box)}\leq\infty$, since the latter is a measure and we have measures on subsets bounded by measures on super-sets. For both cases, we have obtained an upper bound.

We now focus on the first term from \eqref{eqn:proof-5}, rewriting it as:
\beol{\text{Term I}=\sum_{\hat{\ve{x}}\in\scr{S}_l(\Omega)}|S_l(\hat{\ve{x}})|\delta\phi(\hat{\ve{x}})\underbrace{\int_{\lr}\dfrac{\rho_l(\hat{\ve{x}}+l\ve{y})}{K_l(\hat{\ve{x}})}\dm V_{\ve{y}}}_{\sigma_i(\hat{\ve{x}})},}
where $S_l(\hat{\ve{x}})=|\hat{\ve{x}}\oplus l\Box\cap\partial\Omega|$ and $K_l(\hat{\ve{x}})=\dfrac{S_l(\hat{\ve{x}})}{l^2}$; note that $S_l(\hat{\ve{x}})$ is only defined for $\hat{\ve{x}}\in\scr{S}_l(\Omega)$. In the limit as $l\to0$, the above term becomes a surface integral,
\beol{\int_{\partial\Omega}\delta\phi \sigma_{i,h} \dm S_{\ve{s}}}

We now turn to the third term from \eqref{eqn:proof-3} (the outer surface term): 
\begin{equation}
\begin{split}
    &
    \int_{\partial\Omega\times[0,h]}\delta\phi(\ve{s}+lt\ve{n}(\ve{s}))\rho_l(\ve{s}+lt\ve{n}(\ve{s}))\eta(\ve{s}+lt\ve{n}(\ve{s}))\dm t\dm S_{\ve{s}}
    \\
    &
    =
    \underbrace{\int_{\partial\Omega\times[0,h]}[\delta\phi(\ve{s}+lt\ve{n}(\ve{s}))-\delta\phi(\ve{s})]\rho_l(\ve{s}+lt\ve{n}(\ve{s}))\eta(\ve{s}+lt\ve{n}(\ve{s}))\dm t\dm S_{\ve{s}}}_{\text{Term I}}
    +
    \underbrace{\int_{\partial\Omega}\delta\phi(\ve{s})\int_{[0,h]}\rho_l(\ve{s}+lt\ve{n}(\ve{s}))\eta(\ve{s}+lt\ve{n}(\ve{s}))\dm t\dm S_{\ve{s}}}_{\text{Term II}}
\label{eqn:proof-6}
\end{split}
\end{equation}

We consider Term I from \eqref{eqn:proof-6}:
\beol{
    |\text{Term I}|
    \leq 
    \sup_{\Omega^c}\rho_l\int_{\partial\Omega\times[0,h]}~|\eta(\ve{s}+lt\ve{n}(\ve{s}))|~|\delta\phi(\ve{s}+lt\ve{n}(\ve{s}))-\delta\phi(\ve{s})|\dm t\dm S_{\ve{s}}
    \overset{l\downarrow0}{\longrightarrow}0,
}
where the term in the integral goes to $0$ from the continuity of $\delta\phi$.

We consider Term II from \eqref{eqn:proof-6}:
\beol{
    \text{Term II}
    =
    \int_{\partial\Omega}\delta\phi(\ve{s})\underbrace{\int_{[0,h]}\rho_l(\ve{s}+lt\ve{n}(\ve{s}))\eta(\ve{s}+lt\ve{n}(\ve{s}))\dm t}_{\sigma_{o,h}(\ve{s})}\dm S_{\ve{s}}
    =
    \int_{\partial\Omega}\delta\phi(\ve{s})\sigma_{o,h}(\ve{s})\dm S_{\ve{s}}.
}

Now that we have a limit for each term from \eqref{eqn:proof-3}, we can take the limit as $h\to\infty$. 
Using that $\lim_{h\to\infty}G(h)=0$, $\lim_{h\to\infty}\sigma_{o,h}=\sigma_o$ and $\lim_{h\to\infty}\sigma_{i,h}=\sigma_i$, we obtain:
\beol{
    \lim_{l\to0}\dfrac{\delta W^{tot}}{l}
    =
    \delta\tilde{W}^{tot}=\int_{\partial\Omega}\delta\phi(\ve{s})\sigma(\ve{s})\dm S_{\ve{s}}
    +
    \int_{\Omega}\delta\ve{E}\cdot\ve{p}\dm V_{\ve{x}},
}
where $\delta\ve{E}=\nabla\delta\phi$ and $\sigma=\sigma_i+\sigma_o$.

\section{Concluding Remarks}\label{Summary}

In this work, we showed that the three definitions of polarization, namely the classical, the transport, and the energetic definitions of polarization, are equivalent under certain conditions. 
We show that the classical definition of polarization encompasses the transport and the energetic definitions of polarization. 

In Section \ref{cl}, we analyzed the classical definition of polarization as the dipole moment per unit volume over the periodic unit cell.
Owing to the non-uniqueness in the choice of the unit cell, we have non-unique values of dipole moment per unit cell. Thus, the polarization density, a macroscopic description of the electronic nature of the system, depends on the choice of the unit cell in a periodic system.
We showed that different choices of the unit cell in the bulk give rise to different choices of partial unit cells on the surface. The unit cells in the bulk are charge neutral and give rise to a polarization density, while the partial unit cells on the surface are not charge neutral, giving rise to a surface charge density. 
For every choice of unit cell, we have both the polarization density in the bulk as well as the surface charge density on the surface. 
Accounting for both leads to unique potentials and total energies of the system.

In Section \ref{Berry}, we analyze the transport definition of polarization that posits that the change in polarization is the current in the system. 
We showed that the transport definition is related to the classical definition of polarization in that, for infinite systems, the current due to the change in polarization gives rise to the same potential in the system. 
Given that there are assumptions on the transport definition of polarization that are not required for the classical definition of polarization, the classical definition of polarization can be said to encompass the transport definition. 
Further, the transport definition of polarization relies on the adiabatic change or Berry phase associated with any evolving wave function. This concept is manifested as a path integral in Fourier space, which computes a quantized phase related to the Chern number \textbf{[}\cite{vanderbilt_2018}\textbf{]}. From our work reported in this paper, we can conclude that this quantization of the phase introduces a topological aspect to the problem, distinct from the topological aspect of polarization apparent from this approach. Specifically, surface and bulk behaviors exhibit separate characteristics.

Finally, in Section \ref{W}, we analyze the energetic definition of polarization as the energy-conjugate of the electric field. 
Since the variation of the free energy with the electric field is seemingly unique, the polarization so obtained must be unique and independent of the choice of unit cell. 
We showed that analogous to the choice of the bulk and surface unit cells in the classical definition of polarization, the energetic definition of polarization has for every bulk-free energy density a corresponding surface-free energy density; both these contributions are required to obtain a unique total energy.

These conclusions suggest that the classical definition of polarization is the most useful and general.
It is very efficient to compute numerically since it requires only a simple integration over the unit cell.
Further, it allows for treating physical theories at various scales with the same conceptual description.
For instance, in comparison to the transport definition, it does not require the calculation of the quantum mechanical wavefunction, which is an important advantage for the common electronic structure method of Density Functional Theory (DFT).
Despite the calculation of objects that are called wavefunctions in DFT, these are not true wavefunctions since DFT is not truly an \textit{ab initio} electronic structure method.
The basis of DFT is in the electronic charge density, which can be readily used in the classical definition of polarization.
In comparison to the energetic definition, the classical definition is model-independent in that its basis is in the multipole expansion of the fundamental electrostatic equations.
The energetic definition, on the other hand, relies in an essential way on the physical model that is chosen to represent the energy of the system.
That is, given a charge distribution, the corresponding polarization will depend on the physical model that is used to describe the system.

\section*{}

\paragraph*{Competing Interest Statement.}

The authors have no competing interests to declare.

\paragraph*{Acknowledgments.}

We thank Gautam Iyer and Pradeep Sharma for useful discussions;
AFRL for hosting visits by Kaushik Dayal;
and NSF (DMS 2108784), BSF (2018183), and AFOSR (MURI FA9550-18-1-0095) for financial support.
This article draws from the doctoral dissertation of Shoham Sen at Carnegie Mellon University [\cite{Sen2022}].
Shoham Sen was partially supported by a Vannevar Bush Faculty Fellowship (PI: R. D. James) at the University of Minnesota while writing this paper.
	
\appendix

\makeatletter
\renewcommand*{\thesection}{\Alph{section}}
\renewcommand*{\thesubsection}{\thesection.\arabic{subsection}}
\renewcommand*{\p@subsection}{}
\renewcommand*{\thesubsubsection}{\thesubsection.\arabic{subsubsection}}
\renewcommand*{\p@subsubsection}{}
\makeatother

\section{Admissible Set of Unit Cells}\label{Admissible set of unit Cells}

To define a unit cell, we must also define a corner map that goes along with it. We choose ${\Box}^{\prime}\subset\Scr{R}^3$ and $\ve{f}\in\Scr{R}^3$ such that they satisfy 
\beol{
    \Scr{R}^3
    =
    \bigcup_{\ve{i}\in l\scr{L}}\ve{i}\oplus(l\ve{f}\oplus l{\Box}^{\prime}) \text{ and } \forall\ve{i}\neq\ve{j}\in l\scr{L},
    \quad \ve{i}\oplus(l\ve{f}\oplus l{\Box}^{\prime})\cap\ve{j}\oplus(l\ve{f}\oplus l{\Box}^{\prime})=\{\},
\label{2.4}
}
where condition \eqref{2.4} guarantees a pair $(\ve{f},\Box^{\prime})$ which when repeated along the lattice $\scr{L}$ tessellates $\Scr{R}^3$.

The scaled unit cell is defined as $l\Box:=l\ve{f}\oplus l{\Box}^{\prime}$ and the corner map is defined as
\beol{
    \hat{\ve{x}}_l
    :=
    \ve{i}+l\ve{f} \text{ if } \ve{x}\in\Omega \text{ and } \exists\ve{i}\in l\scr{L} \text{ such that } \ve{x}\in\ve{i}\oplus l\Box,
}
where the subscript $l$ in the corner map denotes the scaling of the lattice. We have dropped the subscript $l$ because the scaling is clear from the context.

A scaled partial cell of $l\scr{L}$, denoted $l\lr$, is defined as a scaled unit cell not contained entirely in the domain.
\begin{equation}
    l\lr
    :=
    \left\{\ve{z}\in l\Box\big\vert\exists \ve{i}\in l\scr{L} \text{ such that } \ve{i}+ \ve{z}\in\Omega \text{ but } \ve{i}\oplus l\Box\cap\Omega^c\neq\{\}\right\}.
\end{equation}
The definition of a unit cell and partial unit cell can be obtained from their scaled counterparts by re-scaling each by $l^{-1}$.

\section{Bulk Charge Neutrality}\label{Bulk Charge Neutrality}

\begin{definition}[bulk charge neutrality]
	A sequence of charge distributions $\{\rho_l\}\colon\Scr{R}^3\to\Scr{R}$ is said to be bulk charge neutral on $\Omega(\Box)$, if $\forall \varphi\in L^2(\Scr{R}^3)$, we have
	\beol{\lim_{l\to0}\int_{\Omega(\Box)}\varphi(\ve{x})q_l(\ve{x})\dm V_{\ve{x}}=0,\label{2.7}}
	where the above must hold true independent of the choice of the lattice $l\scr{L}$. Given a choice of the lattice $\scr{L}$, $q_l$ is given by
	\beol{q_l(\ve{x}):=\int_{\Box}\rho_l(\hat{\ve{x}}+l\ve{y})\dm V_{\ve{y}}.\label{2.8}}
	$q_l$ has to be well-defined for all choices of the lattice.
\end{definition}

\begin{remark}
    If our charge distribution varies over space, it is not always possible to ensure that each unit cell will be charge neutral for an arbitrary choice of the lattice.

    A term of the form \eqref{2.7} appears in the derivation of our polarization theorem, and ``bulk charge neutrality'' is a physically-motivated assumptions that gets rid of it. 
    If bulk charge neutrality did not hold, $q_l$ would model the free charge. 
    For modeling ideal dielectrics without free charges, bulk charge neutrality is a reasonable assumption.
\end{remark}


\end{document}

%% file: preamble.tex
\usepackage{isomath}
\usepackage{amsmath,amsthm}
\usepackage{amsbsy}
\usepackage{amssymb}
\usepackage{amscd}
\usepackage{amsfonts}
\usepackage{stmaryrd}
\usepackage{siunitx}
\usepackage{euscript}
\usepackage[utf8]{inputenc}
\usepackage[T1]{fontenc}
\usepackage{newtxtext} 
\everymath{\displaystyle}
\usepackage{exscale}

\usepackage{graphicx}
\usepackage{boxedminipage}
\usepackage{calc}
\usepackage[dvipsnames]{xcolor}
\graphicspath{ {media/} }
\usepackage[caption=false]{subfig}
\usepackage{svg}

\usepackage{setspace}
\usepackage{enumitem}
\setitemize{noitemsep,topsep=0pt,parsep=0pt,partopsep=0pt}
\setenumerate{noitemsep,topsep=0pt,parsep=0pt,partopsep=0pt}
\setdescription{noitemsep,topsep=0pt,parsep=0pt,partopsep=0pt}

\usepackage[normalem]{ulem}

\usepackage[english]{babel}

\usepackage{natbib}
\setcitestyle{authoryear,open={(},close={)}} 

\usepackage[small]{titlesec}

\titlespacing*{\section}{0pt}{12pt plus 4pt minus 2pt}{2pt plus 2pt minus 2pt}
\titlespacing*{\subsection}{0pt}{12pt plus 4pt minus 2pt}{2pt plus 2pt minus 2pt}
\titlespacing*\subsubsection{0pt}{12pt plus 4pt minus 2pt}{2pt plus 2pt minus 2pt}
\titlespacing*\paragraph{0pt}{12pt plus 4pt minus 2pt}{2pt plus 2pt minus 2pt}

\makeatletter
    \renewcommand*{\thesection}{\arabic{section}}
    \renewcommand*{\thesubsection}{\thesection.\Alph{subsection}}
    \renewcommand*{\p@subsection}{}
    \renewcommand*{\thesubsubsection}{\thesubsection.\arabic{subsubsection}}
    \renewcommand*{\p@subsubsection}{}
\makeatother

\input{definitions}

%% file: definitions.tex
\usepackage{isomath}
\usepackage{amsmath}
\usepackage{amssymb}
\usepackage{amscd}
\usepackage{amsfonts}

\newtheorem{theorem}{Theorem}[section]
\newtheorem{lemma}{Lemma}[section]

\theoremstyle{definition}
\newtheorem{definition}{Definition}[section]
\AtEndEnvironment{definition}{\null\hfill\qedsymbol}
\newtheorem{remark}{Remark}[section]
\AtEndEnvironment{remark}{\null\hfill\qedsymbol}
\newtheorem{example}{Example}[section]
\AtEndEnvironment{example}{\null\hfill\qedsymbol}
\newtheorem{assumption}{Assumption}[section]
\AtEndEnvironment{assumption}{\null\hfill\qedsymbol}

\DeclareMathOperator{\divergence}{div}

\newcommand{\dm}{\ \mathrm{d}}

\newcommand{\bfe}{{\mathbold e}}

\newcommand{\bfn}{{\mathbold n}}

\usepackage[utf8]{inputenc}
\usepackage{amsthm}
\usepackage{bbm}
\usepackage{graphicx}
\usepackage{blindtext}
\usepackage{stmaryrd}
\usepackage{color}
\usepackage{lipsum}
\usepackage{tikz}
\usepackage[rightcaption]{sidecap}
\graphicspath{ {images/} }
\usepackage{contour} 
\usepackage{ulem}    
\usepackage{mathtools}

\newcommand{\Scr}[1]{\mathbb{#1}}
\newcommand{\scr}[1]{\mathcal{#1}}

\newcommand{\ve}[1]{\mathbold{#1}}

\newcommand{\TO}[1]{\hat{\textbf{#1}}}
\newcommand{\ket}[1]{\lvert #1\rangle}
\newcommand{\bra}[1]{\langle #1\rvert}

\newcommand{\bracket}[3]{\langle #1\vert #2\vert #3\rangle}

\contourlength{0.8pt}

\newcommand{\beol}[1]{\begin{equation}
#1
\end{equation}}

\newcommand{\bml}[1]{\begin{equation}
\begin{split}
#1
\end{split}
\end{equation}}
\newcommand{\bmlnn}[1]{\begin{equation*}
\begin{split}
#1
\end{split}
\end{equation*}}

\def\lr{\mbox{\begin{picture}(7,10)
\put(1,0){\line(1,0){5}}
\put(1,0){\line(1,2){5}}
\put(6,0){\line(0,1){10}}
\end{picture}
}}